\documentclass[11pt]{article}
\usepackage[latin9]{inputenc}
\usepackage{color}
\usepackage{amsthm}
\usepackage{amsmath}
\usepackage{amssymb}
\usepackage{graphicx}
\usepackage{esint}
\usepackage[square,numbers,sort,compress]{natbib}
\usepackage[colorlinks=true,citecolor=blue]{hyperref}
\PassOptionsToPackage{normalem}{ulem}
\usepackage{calc}
\usepackage{amsthm}
\usepackage{subfigure}
\usepackage{enumerate}
\usepackage{epstopdf}
\usepackage{authblk}
\usepackage[title]{appendix}
\usepackage{color}
\usepackage{ulem}

\theoremstyle{definition}
\newcommand{\vc}{\mathbf}

\newcommand{\bxi}{\pmb\xi}
\newcommand{\id}{\mbox{d}}
\newcommand{\Ff}{\vc F_{t_0}^t}

\newcommand{\R}{\vc R_{t_0}^t}

\newcommand{\Cf}{\vc C_{t_0}^{t}}

\newcommand{\bnabla}{\pmb\nabla}
\newcommand{\tr}{\mathrm{tr}\,}

\newtheorem{theorem}{}

\newtheorem{defn}{Definition}

\newtheorem{prop}[theorem]{Proposition}

\makeatother

\begin{document}
\title{Polar rotation angle identifies elliptic islands in unsteady dynamical systems}
\author[1]{Mohammad Farazmand\thanks{Corresponding author's email address: 
mohammad.farazmand@physics.gatech.edu}}
\author[2]{George Haller}
\affil[1]{\small Center for Nonlinear Science, School of Physics, Georgia Institute of 
Technology, 837 State Street, Atlanta GA 30332, USA }
\affil[2]{\small Department of Mechanical and Process Engineering,
	ETH Z\" urich, 8092 Z\"urich Switzerland}
\date{}
\maketitle
\begin{abstract}
We propose rotation inferred from the polar decomposition of the flow
gradient as a diagnostic for elliptic (or vortex-type) invariant regions
in non-autonomous dynamical systems. We consider here two- and three-dimensional
systems, in which polar rotation can be characterized by a single
angle. For this polar rotation angle (PRA), we derive explicit formulas
using the singular values and vectors of the flow gradient. We find
that closed level sets of the PRA reveal elliptic islands in great
detail, and singular level sets of the PRA uncover centers 
of such islands. Both features turn out to be objective (frame-invariant)
for two-dimensional systems. We illustrate the diagnostic power of
PRA for elliptic structures on several examples.
\end{abstract}

\section{Introduction}

\label{sec:intro} Complex dynamical systems exhibit a mixture of
chaotic and coherent behavior in their phase space. The latter manifests
itself in coherent islands of regular behavior surrounded by a chaotic
background flow. The best known classic examples of such islands are
formed by Kolmogorov--Arnold-Moser (KAM) tori, composed of quasi-periodic
trajectories in Hamiltonian systems \citep[see, e.g.,][]{guckenheimer,topolHydro_arnold}.
Outside elliptic regions filled by such tori, chaotic trajectories
dominate the dynamics.

Even more intriguing is the existence of similar elliptic islands in turbulent
fluid flow, as broadly confirmed by experiments and numerical simulations
\citep[see, e.g.,][]{frischbook,provenzale2008coherent}. Just as
KAM islands, coherent vortices capture trajectories and keep them out
of chaotic mixing zones. Unlike KAM tori, however, coherent vortices
are composed of trajectories that are generally not recurrent in any
frame. During their finite time of existence, these coherent vortices traverse
without filamentation but also without displaying any particular periodic
or quasiperiodic pattern. Still, we generally refer to such regions here
as elliptic, as they mimic the dynamic role of elliptic islands occupied
by classic KAM tori.

Eulerian approaches to describing elliptic islands seek domains where
rotation dominates the instantaneous velocity field. At the simplest
level, this involves locating regions of closed streamlines, high
enough vorticity or low enough pressure (cf. \citet{vortexIdent} and \citet{objVortex}
for reviews). Such domains reveal instantaneous velocity field features
at a low cost, but are unable to frame long-term material coherence
exhibited by trajectories. In addition, the results from these instantaneous
approaches depend on the choice of scalar thresholds and on the frame
of reference.

More sophisticated Eulerian principles for elliptic regions seek sets
of points where rotation dominates strain (see, e.g., 
\citet{okubo,weiss_okubo,hunt1988eddies,hua1998,vortexIdent,tabor1994},
and also \citet{vortexIdent} and \citet{objVortex} for reviews). These principles
infer both rotation and strain from the instantaneous velocity gradient,
thereby rendering the results Galilean invariant. The elliptic regions
they provide, however, still change under rotations of the frame.
Since truly unsteady flows have no distinguished frame of reference
\citep{lugt1979}, frame-dependence in the detection of vortical structures
is an impediment. Indeed, the available measurement velocity data of geophysical flows is 
often given in a rotating frame to begin with, and no
optimal frame is known a priori for structure detection. More importantly,
no mathematical relationship is known (or likely to exist) between
instantaneous rotation-strain principles and material coherence over
extended time intervals.

In contrast, Lagrangian approaches to elliptic islands seek to identify
regions where trajectories stay close for longer periods. These approaches
can roughly be divided into three categories: geometric, set-based
and diagnostic methods. The geometric methods identify elliptic domain
boundaries as spacial closed material lines showing no filamentation
\citep{bhEddy,blazevski_3d,LCS_review} or curvature change \citep{nonhyp_splitting}.
Set-based methods partition the phase space into almost invariant
subsets (see \citet{budivsic2012,froyland} and references therein).
While the boundaries of such sets may undergo filamentation, the overall
subsets remain largely coherent. Finally, diagnostic approaches propose
Lagrangian scalar fields whose features are expected to distinguish
mixing regions from coherent ones 
\citep{mixingFTLE,fsle,rypina2011,Mancho2013,mezic_meso,mundel2014}.
These Lagrangian methods do not return identical results and are not
backed by specific mathematical results on the features they highlight.
In fact, the material invariance of the extracted vortical boundaries is only
guaranteed in the case of the geodesic approach of \citet{bhEddy} and \citet{LCS_review}.

The Lagrangian methods listed above focus on stretching or lack thereof.
In contrast, very few Lagrangian diagnostics target rotation, even
though sustained and coherent rotation is perhaps the most striking
feature of trajectories forming elliptic islands. One of the few exceptions
targeting material rotation is the finite-time rotation number (FTRN),
developed to detect hyperbolic (i.e., repelling or attracting as opposed
to vortical) structures through its ridges \citep{szeceh2013}.
The FTRN assumes that the dynamical system is defined via an iterated
map with an annular phase space. For dynamical systems with general
time dependence and non-annular phase space, however, this approach
is not applicable. This also means that the approach is frame-dependent,
given that translations and rotations will generally destroy the time-periodicity
of a dynamical system. 

Another Lagrangian diagnostic involving a consideration of rotation
is the mesocronic analysis of~\citet{mezic_meso}.
This approach offers a formal extension of the Okubo--Weiss principle
from the velocity gradient to the flow gradient, classifying
an initial condition as elliptic if the flow gradient has complex
eigenvalues at that point. The mesoelliptic diagnostic is efficient
to compute and has been shown to mark vortical regions in several
cases. The direct extension from the Okubo-Weiss principle, however,
also renders the mesoelliptic diagnostic frame-dependent. In addition,
the complex eigenvalues of a finite-time flow map have no known mathematical
relationship with elliptic islands in flows with general time dependence.
Accordingly, some annular subsets of classic elliptic domains fail
the test of meso-ellipticity even in steady flows (cf.~\citep{mezic_meso}, Fig.~1).

Here we propose a mathematically precise assessment of material rotation, the
polar rotation angle (PRA), as a new diagnostic for elliptic islands
in two- and three-dimensional flows. The PRA is the angle of the rigid-body
rotation component obtained from the classic polar decomposition of
the flow gradient into a rotational and a stretching factor.
We show how the PRA can readily be computed from invariants of the
flow gradient and the Cauchy--Green strain tensor. Level sets
of the PRA turn out to be objective (frame-invariant) in planar
flows. We find that these level sets reveal the internal structure
of elliptic islands in great detail at a relatively low computational
cost. We also find that local extrema of the PRA mark elliptic island
centers suitable for automated vortex tracking in Lagrangian fluid dynamics.

\section{Preliminaries}

\label{sec:prelim}

\subsection{Set-up}

Consider the dynamical system 
\begin{eqnarray}
\dot{\vc x}=\vc u(\vc x,t),\ \ \ \vc x\in \mathcal D\subset\mathbb{R}^{3},\ \ \ t\in 
I\subset\mathbb{R},\label{eq:dynsys}
\end{eqnarray}
with the corresponding flow map 
\begin{align}
\Ff:\  & \mathcal D\rightarrow \mathcal D\nonumber \\
& \vc x_{0}\mapsto\vc x(t;t_{0},\vc x_{0}),\label{eq:flowMap}
\end{align}
the diffeomorphism that takes the initial condition $\vc x_{0}$ to
its time-$t$ position $\vc x(t;t_{0},\vc x_{0})$ under system \eqref{eq:dynsys}.
Here, $\mathcal D$ denotes the phase space and $I$ is a finite time interval
of interest.

The deformation gradient $\bnabla\Ff$ governs the infinitesimal deformations
of the phase space $\mathcal D$. In particular, an initial perturbation $\bxi$
at point $\vc x_{0}$ and time $t_{0}$ is mapped, under the system
\eqref{eq:dynsys}, to $\bnabla\Ff(\vc x_{0})\bxi$ at time $t$.
We also define the \emph{Cauchy--Green strain tensor}, 
\begin{equation}
\Cf:=\left[\bnabla\Ff\right]^{\top}\bnabla\Ff:\vc x_{0}\mapsto\Cf(\vc x_{0}),
\label{eq:CG}
\end{equation}
where the symbol $\top$ denotes matrix transposition. The tensor
$\Cf(\vc x_{0})$ is symmetric and positive definite. Therefore, it
has an orthonormal set of eigenvectors $\{\bxi_{1}(\vc x_{0}),\bxi_2(\vc 
x_{0}),\bxi_{3}(\vc x_{0})\}$.
The corresponding eigenvalues $0<\lambda_{1}(\vc x_{0})\leq\lambda_{2}(\vc 
x_{0})\leq\lambda_{3}(\vc x_{0})$
therefore satisfy 
\begin{equation}
\Cf(\vc x_{0})\bxi_{i}(\vc x_{0})=\lambda_{i}(\vc x_{0})\bxi_{i}(\vc x_{0}),\ \ \ 
i\in\{1,2,3\},
\end{equation}
\begin{equation}
\langle\bxi_{j}(\vc x_{0}),\bxi_{k}(\vc x_{0})\rangle=0,\ \ \ j,k\in\{1,2,3\},\ \ \ j\neq 
k,
\end{equation}
with $\langle\cdot,\cdot\rangle$ denoting the Euclidean inner product.
For notational simplicity, we omit the dependence of the eigenvalues
and eigenvectors on $t_{0}$ and $t$. Also, we consider two-dimensional
flows as a special case satisfying $\partial_{x_{3}}\vc u_{i}(\vc x,t)\equiv0,$
$i=1,2,3$.

\subsection{Polar decomposition}

\label{sec:polarDec} Any square matrix admits a factorization into
the product of a unitary matrix with a symmetric positive-semidefinite
matrix \citep{conwayOperatorTheory}. When the square matrix is nonsingular,
such as $\bnabla\Ff$, then the symmetric factor in the decomposition
is positive definite. 

Specifically, the deformation gradient $\bnabla\Ff$ admits a unique
decomposition of the form 
\begin{equation}
\bnabla\Ff=\R\mathbf{U}_{t_{0}}^{t},\label{eq:polarDec}
\end{equation}
where the $3\times3$ matrices $\R$ and $\mathbf{U}_{t_{0}}^{t}$
have the following properties \citep{conwayOperatorTheory,gurtin1982,Truesdell09}: 
\begin{enumerate}
	\item The rotation tensor $\R$ is proper orthogonal, i.e., 
	\[
	\big(\R\big)^{\top}\R=\R\big(\R\big)^{\top}=\mathbf{I}\mbox{\mbox{,}}\qquad\det\R=1.
	\]
	
	\item The right stretch tensor $\mathbf{U}_{t_{0}}^{t}$ is symmetric and
	positive-definite, satisfying 
	\begin{equation}
	\left[\mathbf{U}_{t_{0}}^{t}\right]^{2}=\Cf.\label{eq:square}
	\end{equation}
	
	\item The eigenvalues of $\mathbf{U}_{t_{0}}^{t}$ are $\sqrt{\lambda_{k}}$
	with corresponding eigenvectors $\bxi_{k}$: 
	\begin{align}
	\mathbf{U}_{t_{0}}^{t}(\vc x_{0})\bxi_{k}(\vc x_{0}) & =\sqrt{\lambda_{k}(\vc 
	x_{0})}\bxi_{k}(\vc x_{0}),\ \ \ k=1,2,3,\label{eq:rescaling}\\
	\nonumber 
	\end{align}
	
	\item The time derivative of the rotation tensor satisfies 
	\begin{equation}
	\dot{\mathbf{R}}_{t_{0}}^{t}=\left(\mathbf{W}\left(\vc x(t),t\right)-
	\frac{1}{2}\mathbf{R}_{t_{0}}^{t}
	\left[\dot{\mathbf{U}}_{t_{0}}^{t}
	\left(\mathbf{U}_{t_{0}}^{t}\right)^{-1}-
	\left(\mathbf{U}_{t_{0}}^{t}\right)^{-1}
	\dot{\mathbf{U}}_{t_{0}}^{t}\right]
	\left(\mathbf{R}_{t_{0}}^{t}\right)^\top\right)\mathbf{R}_{t_{0}}^{t},
	\label{eq:derivR}
	\end{equation}
	where
	$\mathbf{W}=\frac{1}{2}\left[\nabla\mathbf{u}-
	\left(\nabla\mathbf{u}\right)^{\top}\right]$
	is the vorticity (or spin) tensor and $\vc x(t)$ is a shorthand notation for the 
	trajectory $\vc x(t;t_{0},\vc x_{0})$. A derivation of \eqref{eq:derivR} can be 
	found, 
	e.g., in~\cite[][Section 23]{truesdell}.
\end{enumerate}

The geometric interpretation of the polar decomposition is the following
\citep{gurtin1982,introContinuum}. At any point $\vc x_{0}$ of the
phase space, the orthogonal basis $\{\bxi_{k}\}_{1\leq k\leq3}$ is mapped
into $\{\bnabla\Ff(\vc x_{0})\bxi_{k}\}_{1\leq k\leq3}$ under the linearized
flow map $\bnabla\Ff$. Any stretching and compression in the deformation
is encoded into the stretch tensor $\mathbf{U}_{t_{0}}^{t}$, while
the overall rigid-body rotation of material elements is encoded into
the rotation tensor $\R$. Figure~\ref{fig:polarDec} illustrates the
action of these tensors on area elements in two dimensions. 
\begin{figure}[t!]
	\centering \includegraphics[width=\textwidth]{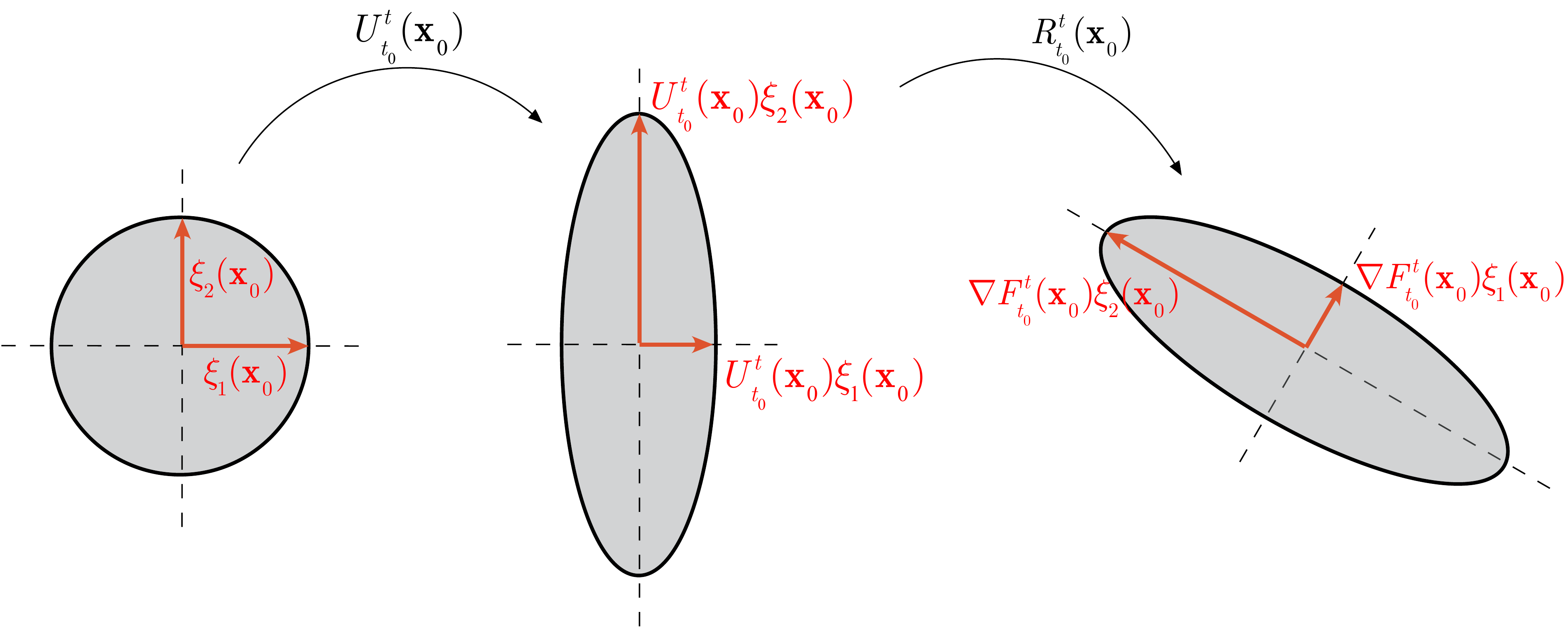}
	\protect\caption{The action of the deformation gradient $\bnabla\Ff$ is uniquely 
	decomposable
		into positive definite stretch by $\mathbf{U}_{t_{0}}^{t}$ followed
		by rotation by $\R$. This results in the polar decomposition 
		$\bnabla\Ff=\R\mathbf{U}_{t_{0}}^{t}$.}
	\label{fig:polarDec} 
\end{figure}

When mapped forward under the deformation gradient $\nabla\Ff$, a
general unit vector $\mathbf{a}$ experiences two stages of rotation.
First, $\mathbf{a}$ is rotated (and simultaneously stretched) by
the stretch tensor $\mathbf{U}_{t_{0}}^{t}$ into the vector 
$\mathbf{U}_{t_{0}}^{t}\mathbf{a}$.
This first stage of rotation is entirely due to shear, with the magnitude
and axis of rotation depending on $\mathbf{a}$. The second stage
of rotation experienced by $\mathbf{a}$ is due to the rotation tensor
$\mathbf{R}_{t_{0}}^{t}$, which rotates $\mathbf{a}$ into its final
position $\mathbf{R}_{t_{0}}^{t}\mathbf{U}_{t_{0}}^{t}\mathbf{a}$
at time $t$. This second rotation acts in the same way on all
$\mathbf{U}_{t_{0}}^{t}\mathbf{a}$ vectors by the proper orthogonal
nature of $\mathbf{R}_{t_{0}}^{t}$. 

Formed by the eigenvectors of $\mathbf{C}_{t_{0}}^{t}$, the principle
rectangle illustrated in Fig.~\ref{fig:lagRot} has a special feature:
it is the unique rectangle on which the first stage of rotation under
$\mathbf{U}_{t_{0}}^{t}$ is inactive. This is because the edges of
the principal rectangle align with the eigenvectors of $\mathbf{U}_{t_{0}}^{t}$
(cf. eq.~\eqref{eq:rescaling} and Fig.~\ref{fig:polarDec}), and hence remain 
unrotated by 
$\mathbf{U}_{t_{0}}^{t}$.
The total rotation experienced by the edges of the principal rectangle
is, therefore, just the rigid-body rotation exerted by the rotation
tensor $\mathbf{R}_{t_{0}}^{t}$. 

\begin{figure}[t!]
	\includegraphics[width=1\textwidth]{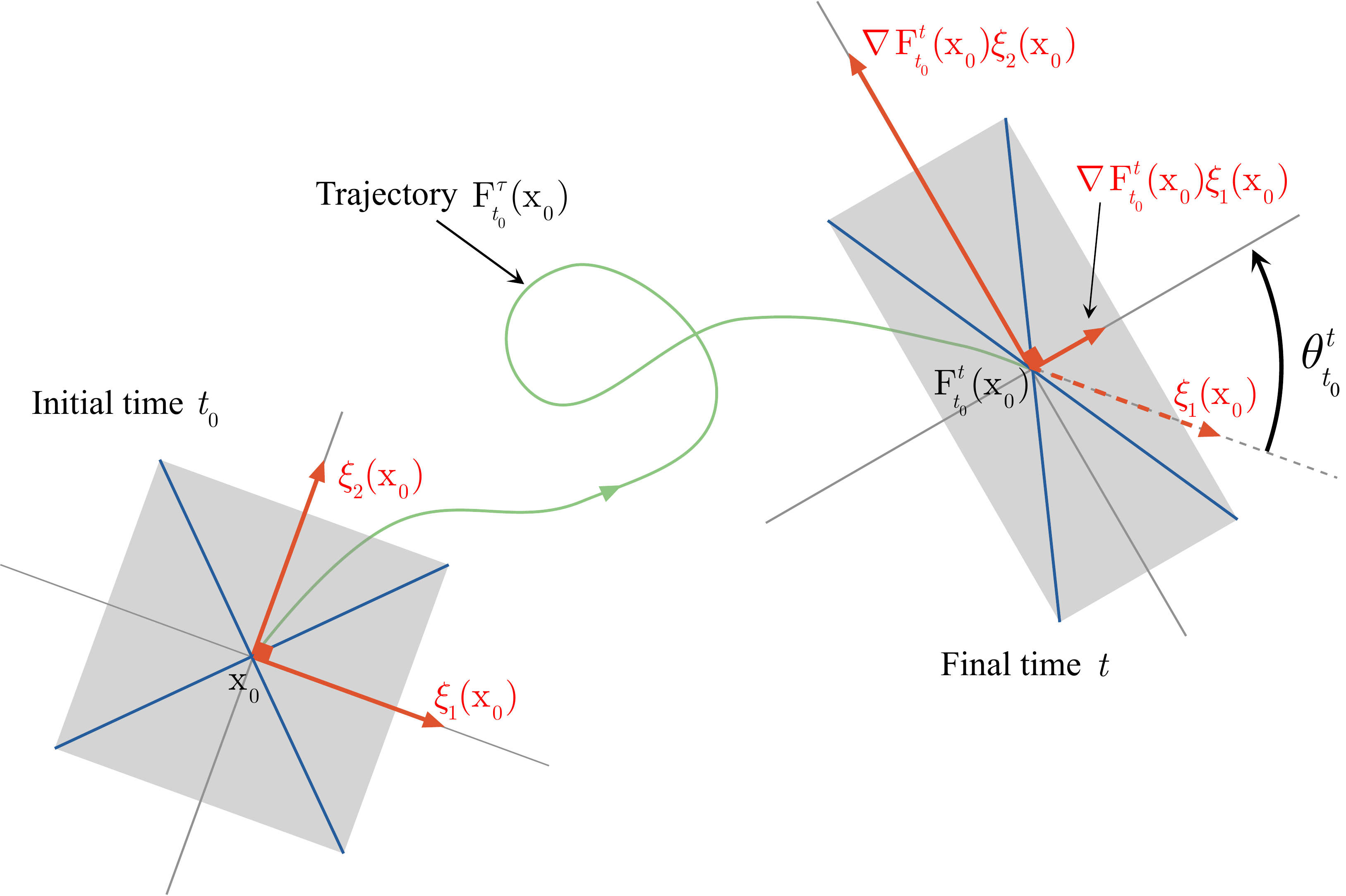}
	\caption{Finite-time deformation of an area element of a two-dimensional phase
		space under the flow map $\vc F_{t_{0}}^{t}$. The orthonormal basis 
		$\{\bxi_1,\bxi_2\}$ 
		is mapped to the orthogonal basis $\{\bnabla\Ff\bxi_1,\bnabla\Ff\bxi_2\}$. Other 
		initially orthogonal material elements, such as the diagonals shown in blue, are 
		mapped 
		to non-orthogonal material elements.}
	\label{fig:lagRot} 
\end{figure}

Formula \eqref{eq:derivR} shows the difference between instantaneous
Eulerian rotation measured by the vorticity tensor $\mathbf{W}$ and
finite-time material rotation measured by $\R$. In particular, at an initial time 
$t_{0}$, we
have 
$$\dot{\mathbf{R}}_{t_{0}}^{t}\big|_{t={t_{0}}}=\mathbf{W}(\mathbf{x}_{0},t_{0}),$$
but $\dot{\mathbf{R}}_{t_{0}}^{t}$
differs from the vorticity tensor $\vc W(\vc x(t),t)$ for times $t\neq t_{0}$.

\section{Polar rotation angle (PRA)}
The classic procedure for computing the polar decomposition in continuum
mechanics starts with determining $\mathbf{U}_{t_{0}}^{t}$as the
principal square root of the Cauchy--Green strain tensor (cf. formula
\eqref{eq:square}). This is the simplest to do by diagonalizing $\mathbf{C}_{t_{0}}^{t}$,
taking the positive square root of its diagonal elements, and transforming
back the resulting matrix from the strain eigenbasis to the original
basis. Next, one obtains the rotation tensor directly from \eqref{eq:polarDec}
as $\mathbf{R}_{t_{0}}^{t}=\nabla\Ff\left[\mathbf{U}_{t_{0}}^{t}\right]^{-1}$. More
efficient numerical procedures are also available (see \citep{Bouby05}
and the references cited therein)

These computational approaches, however, offer little insight into
the geometry of the rotation generated by $\mathbf{R}_{t_{0}}^{t}$.
Taking a more geometric approach, one may recall that any three-dimensional
rotation $\mathbf{R}_{t_{0}}^{t}$ has a Rodrigues representation
\citep{NonlinContinuumMech} of the form 
\begin{equation}
\R=\mathbf{I}+\sin\theta_{t_{0}}^{t}\mathbf{P}_{t_{0}}^{t}+
(1-\cos\theta_{t_{0}}^{t})\left[\mathbf{P}_{t_{0}}^{t}\right]^{2},
\label{rotMat_theta}
\end{equation} 
where \textbf{$\mathbf{I}$} is the $3\times3$ identity matrix and
$\mathbf{P}_{t_0}^t$ is a $3\times3$ skew-symmetric matrix such that
\[
\mathbf{P}_{t_{0}}^{t}\mathbf{a}=\mathbf{r}_{t_{0}}^{t}\times\mathbf{a},\quad\forall
\mathbf a\in\mathbb R^3.
\]
The unit vector $\mathbf r_{t_0}^t$ is the eigenvector of $\R$ corresponding to its
unit eigenvalue, i.e.,
\begin{equation}
\mathbf{R}_{t_{0}}^{t}(\mathbf{x}_{0})\mathbf{r}_{t_{0}}^{t}(\mathbf{x}_{0})=
\mathbf{r}_{t_{0}}^{t}(\mathbf{x}_{0}).
\label{eq:axisRot}
\end{equation}
For planar flows, the eigenvector $\mathbf r_{t_0}^t$ is the unit normal 
to the plane of motion, and hence is independent of $\vc x_0$. In three-dimensions, $\vc 
r^t_{t_0}$ depends on the location $\vc x_0$ in a way discussed in the next section (cf. 
Proposition \ref{prop:pra})

Once an orientation for the unit vector 
$\mathbf{r}_{t_{0}}^{t}(\mathbf{x}_{0})$ is selected,
the angle $\theta_{t_{0}}^{t}(\mathbf{x}_{0})\in[0,2\pi)$ is uniquely
determined. This angle measures the amount of local
solid--body rotation experienced by material elements along the trajectory
$\vc x(t;t_{0},\vc x_{0})$.
\begin{defn}
	We refer to the scalar function
	\[
	\theta_{t_{0}}^{t}(\vc x_{0})\in[0,2\pi)
	\]
	determined by \eqref{rotMat_theta} as the \emph{polar rotation angle}
	(PRA) at the initial condition $\vc x_{0}$ with respect to the time
	interval $[t_{0},t]$.
	\label{def:PRA}
\end{defn}

\section{Computing the PRA}

Taking the trace of both sides in \eqref{rotMat_theta}, then
taking the skew-symmetric part of both sides of \eqref{rotMat_theta}
yields the formulas 
\begin{subequations}
	\begin{align}
	\cos\theta_{t_{0}}^{t} & =\frac{1}{2}\left(\tr\R-1\right),
	\label{eq:LagRot_3d_cos} \\
	\sin\theta_{t_{0}}^{t} & =\frac{\left[\hat{\mathbf{R}}_{t_{0}}^{t}\right]_{ij}}
	{\left[\mathbf{P}_{t_{0}}^{t}\right]_{ij}}
	\quad (i\neq j),\qquad\hat{\mathbf{R}}_{t_{0}}^{t}:=
	\frac{1}{2}\left(\mathbf{R}_{t_{0}}^{t}-\left[\mathbf{R}_{t_{0}}^{t}\right]^{T}\right).
	\label{eq:LagRot_3d_sin}
	\end{align}
	\label{eq:LagRot_3d}
\end{subequations}
To evaluate the expression for $\cos\theta_{t_{0}}^{t}$ in \eqref{eq:LagRot_3d_cos}, 
Guan-Suo \citep{guan-suo98}
expressed $\tr\R$ as a somewhat cryptic function of the scalar invariants
of the matrices $\nabla\Ff$, 
$\frac{1}{2}\left(\nabla\Ff+\left[\nabla\Ff\right]{}^{\top}\right)$
and $\mathbf{U}_{t_{0}}^{t}$. Here we derive a simply computable
and intuitive alternative that only involves quantities arising in
typical Lagrangian coherent structure calculations~\cite{LCS_review}: the deformation
gradient, and the eigenvalues and eigenvectors of the the Cauchy--Green
strain tensor.
\begin{prop}\leavevmode
	\begin{enumerate}[(1)]
		\item In three-dimensional flows, the PRA satisfies the relations
		\begin{subequations}
			\label{eq:PRA3D}
			\begin{align}
			\cos\theta_{t_{0}}^{t} & =\frac{1}{2}\left(\sum_{i=1}^{3}\frac{\left\langle 
				\bxi_{i},\nabla\Ff\bxi_{i}\right\rangle }{\sqrt{\lambda_{i}}}-1\right),
			\label{eq:cosPRA_3d} \\
			\sin\theta_{t_{0}}^{t} & =\frac{\left\langle 
			\bxi_{i},\nabla\Ff\bxi_{j}\right\rangle 
				-\left\langle \bxi_{j},\nabla\Ff\bxi_{i}\right\rangle 
				}{2\epsilon_{ijk}e_{k}},\qquad 
			i\neq j\in\{1,2,3\},
			\label{eq:sinPRA_3d}
			\end{align}
		\end{subequations}
		where $\mathbf{e}=(e_1,e_2,e_3)^\top$ is the normalized eigenvector corresponding
		to the unit eigenvalue of the matrix
		$$\left[\mathbf{K}_{t_{0}}^{t}\right]_{jk}=\frac{\left\langle 
			\bxi_{j},\nabla\Ff\bxi_{k}\right\rangle }{\sqrt{\lambda_{k}}},
		\quad j,k\in\{1,2,3\},$$
		and $\epsilon_{ijk}$ is the Levi-Civita symbol. Furthermore, we have $\vc 
		e = \vc r_{t_0}^t$ where $\vc r_{t_0}^t$ is the axis of rotation defined by 
		\eqref{eq:axisRot}.
		
		\item In two-dimensional flows, we have 
		\begin{subequations}
			\begin{align}
			\cos\theta_{t_{0}}^{t} & =\frac{\left\langle 
			\bxi_{i},\nabla\Ff\bxi_{i}\right\rangle 
			}{\sqrt{\lambda_{i}}},\qquad i=1\;\mbox{or}\;\,2, 
			\label{eq:cosPRA}\\
			\sin\theta_{t_{0}}^{t} & =\left(-1\right)^{j}\frac{\left\langle 
				\bxi_{i},\nabla\Ff\bxi_{j}\right\rangle 
			}{\sqrt{\lambda_{j}}},\qquad 
			\left(i,j\right)=\ensuremath{\left(1,2\right)}\quad 
			\mbox{or}\quad\ensuremath{\left(2,1\right)},
			\label{eq:sinPRA}
			\end{align}
			\label{eq:PRA2D}
		\end{subequations}
		where $\lambda_{1}\leq\lambda_{2}$ are the eigenvalues of the two-dimensional
		Cauchy--Green strain tensor with corresponding eigenvectors $\bxi_{1}$
		and $\bxi_{2}$.
	\end{enumerate}
	\label{prop:pra}
\end{prop}

\begin{proof}
	See Appendix \ref{app:proof_prop1}. 
\end{proof}

Using both expressions in the formulas \eqref{eq:PRA3D} (or formulas \eqref{eq:PRA2D}, in 
the two-dimensional case), the four-quadrant 
polar rotation angle $\theta_{t_{0}}^{t}\in[0,2\pi)$ can be reconstructed
as
\begin{equation}
\theta_{t_{0}}^{t} =\left[1-{\rm 
	sign\,}\left(\sin\theta_{t_{0}}^{t}\right)\right]\pi+{\rm 
	sign\,}\left(\sin\theta_{t_{0}}^{t}\right)\cos^{-1}
\left(\cos\theta_{t_{0}}^{t}\right),
\label{eq:fullpolar}
\end{equation}
where
\begin{equation*}
\rm{sign}(\alpha) = \left\{
\begin{array}{l l}
1  &    \qquad       \rm{if}\quad \alpha\geq 0\\
-1 &    \qquad       \rm{if}\quad \alpha< 0
\end{array}	\right.
\end{equation*}

For completeness, in Appendix \ref{app:total_rot}, we also derive a formula for the
total rotation of an arbitrary material element, not just for the
strain eigenvectors. Evaluating this general formula is computationally
more costly, as it involves advecting initial directions by the flow
map through all intermediate times within the interval $[t_{0},t]$.
In addition, due to the non-rigid-body nature of deformation along
a trajectory, the total material rotation will be different for different
material elements. When evaluated on initial directions aligned with
$\bxi_{1}$ and $\bxi_{2}$, however, this total Lagrangian rotation
agrees with the PRA modulo multiples of $2\pi$. 

\section{Polar LCS}
A recent approach to the systematic detection of elliptic Lagrangian
coherent structures (LCS) targets closed material lines that exhibit
no filamentation over the finite time interval $[t_{0},t]$ (\citet{bhEddy,LCS_review}).
These elliptic LCSs turn out to be uniformly stretching closed material
lines, i.e., all their subsets exhibit the same relative stretching.
Outermost members of nested elliptic LCS families then serve as the
ideal boundaries of perfectly coherent elliptic islands.

Here we propose a dual approach to elliptic LCSs by requiring uniformity
in the polar rotation of material elements forming the LCS, as opposed
to uniformity in their stretching. 
\begin{defn}
	A \emph{polar Lagrangian coherent structure} (polar LCS) over the
	time interval $[t_{0},t]$ is a closed (i.e., tubular in 3D and circular
	in 2D) and connected codimension-one material surface whose time $t_{0}$
	position is a level set of $\theta_{t_{0}}^{t}(\vc x_{0})$.
	\label{def:polarLCS}
\end{defn}
As any material surface, a polar LCS is invariant under the flow.
It is formed by trajectories starting from a closed and connected
level set of $\theta_{t_{0}}^{t}(\vc x_{0})$ at time $t_{0}$. The
following simple observation shows that polar LCSs can be detected
as connected and closed level sets of trigonometric functions of $\theta_{t_{0}}^{t}(\vc 
x_{0})$,
and hence are directly computable from the formulas \eqref{eq:PRA3D}-\eqref{eq:PRA2D}.
\begin{prop}
	Connected components of the level sets of $\cos\theta_{t_{0}}^{t}$
	and $\sin\theta_{t_{0}}^{t}$coincide with connected components of
	the level sets of $\theta_{t_{0}}^{t}(\vc x_{0})$.
	\label{prop:connComp}
\end{prop}
\begin{proof}
	Assume the contrary, i.e., the existence of two points $\mathbf{x}_{0}$
	and $\mathbf{\hat{x}}_{0}$ that are in the same connected component
	of a level set $\mathcal{L}$ of $\cos\theta_{t_{0}}^{t}(\vc x_{0})$
	but on different connected level sets of $\theta_{t_{0}}^{t}(\vc x_{0})$.
	Then on any continuous path connecting $\mathbf{x}_{0}$ and $\mathbf{\hat{x}}_{0}$,
	the polar rotation angle $\theta_{t_{0}}^{t}$ should change continuously
	from $\theta_{t_{0}}^{t}(\mathbf{x}_{0})$ to 
	$\theta_{t_{0}}^{t}(\mathbf{\hat{x}}_{0})\neq\theta_{t_{0}}^{t}(\mathbf{x}_{0}),$
	and hence $\cos\theta_{t_{0}}^{t}$ cannot be constant along this
	path. Since $\mathbf{x}_{0}$ and $\mathbf{\hat{x}}_{0}$ are in the
	connected set $\mathcal{L}$, there is therefore a continuous path
	connecting $\mathbf{x}_{0}$ and $\mathbf{\hat{x}}_{0}$ within $\mathcal{L}$
	along which $\cos\theta_{t_{0}}^{t}(\vc x_{0})$ cannot be constant.
	But this contradicts the assumption that $\mathcal{L}$ is a level
	set of $\cos\theta_{t_{0}}^{t}(\vc x_{0})$. The argument for $\sin\theta_{t_{0}}^{t}$
	is identical.
\end{proof}

A practical consequence of Proposition \ref{prop:connComp} is that connected level sets of
$\theta_{t_{0}}^{t}(\vc x_{0})$ can be constructed as those of $\cos\theta_{t_{0}}^{t}$
and $\sin\theta_{t_{0}}^{t}$, without verifying the orientability
of $\mathbf{r}_{t_{0}}^{t}(\mathbf{x}_{0})$ on $U$. 
This renders the computation of the tensor $\vc 
K_{t_0}^t$ and the rotation axis $\vc r_{t_0}^t$ unnecessary, as one can 
compute the two-quadrant angle $\theta_{t_0}^t\in[0,\pi]$ from equation 
\eqref{eq:cosPRA_3d} as
\begin{equation}
\theta_{t_0}^t=\cos^{-1}\left[\frac{1}{2}\left(\sum_{i=1}^{3}\frac{\left\langle 
	\bxi_{i},\nabla\Ff\bxi_{i}\right\rangle }{\sqrt{\lambda_{i}}}-1\right)\right].
\label{eq:th_2q}
\end{equation}
Proposition \ref{prop:connComp} ensures that the level sets of $\theta^t_{t_0}$ computed 
from \eqref{eq:th_2q} coincide with those of the four-quadrant PRA angle computed from 
\eqref{eq:fullpolar}.

All quantities derived from the deformation gradient $\bnabla\Ff$ are invariant
with respect to time-dependent translations of the coordinate frame.
Therefore, polar LCSs are Galilean invariant objects. For two-dimensional
flows, polar LCSs also turn out to be invariant under time-dependent
rotations of the frame. In the language of continuum mechanics \citep{truesdell},
polar LCSs in two dimensions are objective.
\begin{prop}
	In two-dimensional flows, a polar LCS over the time interval $[t_{0},t]$
	is objective, i.e., invariant under coordinate changes of the form
	\begin{equation}
	\mathbf{x}=\mathbf{Q}(t)\mathbf{y}+\mathbf{b}(t),
	\label{eq:Eucl}
	\end{equation}
	where $\mathbf{Q}(t)\in SO(2)$ and $\mathbf{b}(t)\in\mathbb{R}^{2}$
	are smooth functions of time $t$.
	\label{prop:obj}
\end{prop}

\begin{proof}
	See Appendix \ref{app:proof_prop3}.
\end{proof}
An elliptic island marked by the PRA has a natural center point: the
PRA extremum point surrounded by closed PRA contours. This leads to
the following definition of a Lagrangian vortex center:
\begin{defn}
	A \emph{Lagrangian vortex center} over the time interval $[t_{0},t]$
	is a set of trajectories evolving from a connected, codimension-two
	level set of $\theta_{t_{0}}^{t}(\vc x_{0})$. 
	\label{def:vc}
\end{defn}
A Lagrangian vortex center identified from the PRA is, therefore,
composed of a single trajectory in two dimensions and of a one-parameter
family of trajectories (i.e., a material line) in three dimensions.
Despite recent progress in the accurate detection of coherent Lagrangian
vortex boundaries \citep{LCS_review}, approaches to Lagrangian vortex
center definition and detection have notably been missing. As we illustrate
in Section \ref{sec:2Dturb} below, vortex centers defined as PRA extrema indeed
show distinguished behavior: they capture the translational motion
of an elliptic island without being affected by the rotational motion
of trajectories inside the island. As connected level sets of the
PRA, the Lagrangian vortex centers defined in Definition \ref{def:vc} are also
objective in two-dimensional flows (cf. Proposition \ref{prop:obj}).

\section{Examples}

\label{sec:examples} In this section, we compute the PRA on several
examples to illustrate how its closed level curves (i.e., initial
positions of polar LCSs) highlight the internal structure of elliptic
islands in detail.

\subsection{Standard map}\label{sec:standMap} 
We first consider the standard map 
\begin{align}
I_{n+1} & =I_{n}+\epsilon\sin\phi_{n},\nonumber \\
\phi_{n+1} & =\phi_{n}+I_{n+1},\label{eq:standardMap}
\end{align}
which is a Poincaré map $\mathcal{P}$ of a rotor excited by a periodic
impulsive force~\citep{ott}. In the absence of the impulse, i.e.,
for $\epsilon=0$, the angular momentum $I_{n}$ is constant and the
angular position $\phi_{n}$ increases linearly as an integer multiple
of the angular momentum.

\begin{figure}[t!]
	\centering 
	\includegraphics[width=0.95\textwidth]{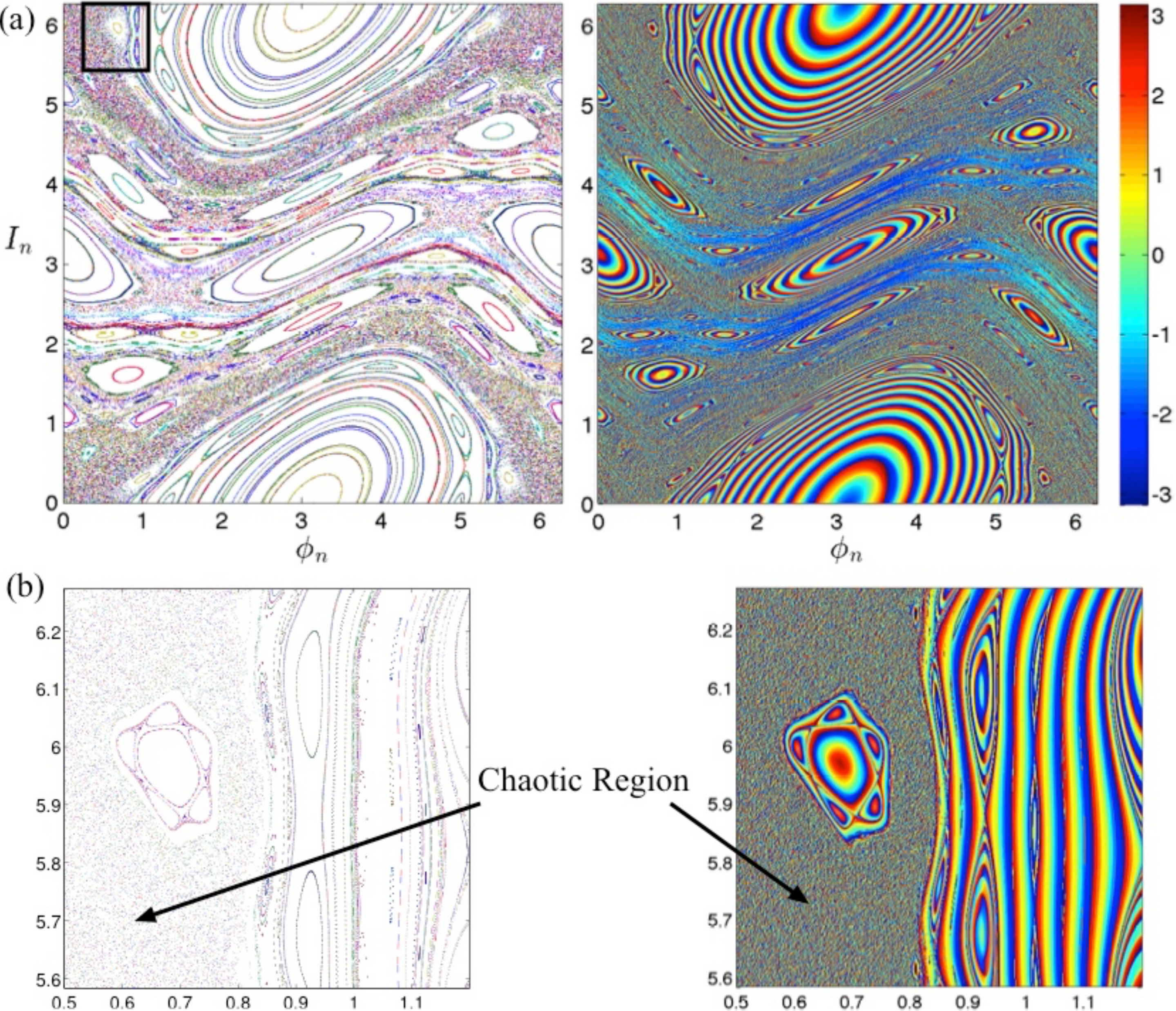}
	\caption{Standard map. (a) Left: $1000$ iterations of the standard map for
		$400$ uniformly distributed initial conditions over the torus 
		$[0,2\pi]\times[0,2\pi]$. 
		Right: The PRA $\theta_{t_0}^t$ for $200$ iterations of the standard map, clearly 
		marking 
		polar LCSs
		(closed contours) and Lagrangian vortex centers (local extrema) (b)
		The close-up view of the region marked by a rectangle in (a).}
	\label{fig:standardMap}
\end{figure}

For $\epsilon\neq0$, however, the system can exhibit complicated
dynamics. Depending on the initial condition $(I_{0},\phi_{0})$,
the trajectories may be periodic, quasi-periodic or chaotic. The quasi-periodic
trajectories lie on KAM tori, the classic examples of vortical structures
that we wish to visualize through the PRA. 

The left plot in Fig.~\ref{fig:standardMap}a shows $1000$ iterations
of the standard map for $400$ uniformly distributed initial conditions
and $\epsilon=1$. This reveals invariant KAM tori, resonance islands
and chaotic regions. The right panel of the same figure shows the PRA,
computed from formula \eqref{eq:fullpolar} with $i=2$, with the flow
map being equal to $200$ iterations of the map \eqref{eq:standardMap},
i.e. $\Ff(\vc x_{0})=\mathcal{P}^{200}(\vc x_{0})$, 
where $\vc x_{0}=(\phi_{0},I_{0})$
and $(\phi_{n+1},I_{n+1})=\mathcal{P}(\phi_{n},I_{n})$. To ensure
the accuracy of the finite differences for the computation of the
deformation gradient $\bnabla\Ff$, we use a dense grid of initial
conditions consisting of $1000\times1000$ uniformly distributed points over the phase 
space $\mathbb T^2=[0,2\pi]\times[0,2\pi]$.

Figure~\ref{fig:standardMap}b shows a close-up of a region of the phase
space containing chaotic trajectories, KAM tori and a period-$5$
resonance island. For this close-up view, the Poincaré map is recomputed
from $1000$ iterations of $2500$ initial conditions. The corresponding
PRA plot on the right is computed only from $500$ iterations, i.e.,
from $\Ff=\mathcal{P}^{500}$.

We conclude that the KAM tori and resonance islands are sharply enhanced
by the PRA relative to a simple iteration of the map, even though
the number of iterations used in constructing the PRA plot is only
half the number used for the Poincar\'e map. The chaotic region is
marked by small-scale rapid variations in the PRA, in line with the
sensitive dependence of the rotation angle on initial conditions in these regions. 

Figure~\ref{fig:standardMap} also shows that the center-type fixed
points in the elliptic islands are clearly marked with local extrema
of the PRA, supporting the idea of defining Lagrangian elliptic centers
as stated in Definition~\ref{def:vc}.

\subsection{Two-dimensional turbulence}\label{sec:2Dturb} 
Consider the Navier--Stokes equation 
\begin{equation}
\partial_{t}\vc u+\vc u\cdot\vc\bnabla\vc u=-\bnabla p+\nu\Delta\vc u+\vc f,\ \ \ 
\bnabla\cdot\vc u=0,\label{eq:nse}
\end{equation}
where $\nu$ is the kinematic viscosity and $\vc f$ denotes the forcing.
For an ideal two-dimensional fluid flow ($\nu=0$ and $\vc f=\vc0$
), the vorticity $\omega$, given by $\bnabla\times\vc u=(0,0,\omega)$,
is preserved along fluid trajectories, i.e., 
\begin{equation}
\frac{\mathrm{D}\omega}{\mathrm{D}t}=0,\label{eq:vort_pres}
\end{equation}
where $\frac{\mathrm{D}\;}{\mathrm{D}t}:=\partial_{t}+\vc u\cdot\bnabla$
is the material derivative. Therefore, closed level curves of vorticity
are material curves, acting as barriers to the transport of fluid
particles. In the presence of molecular diffusion and external forcing,
however, vorticity is not a material invariant and hence its closed contours
no longer signal elliptic islands for fluid trajectories.

\begin{figure}[t!]
	\centering 
	\includegraphics[width=0.46\textwidth]{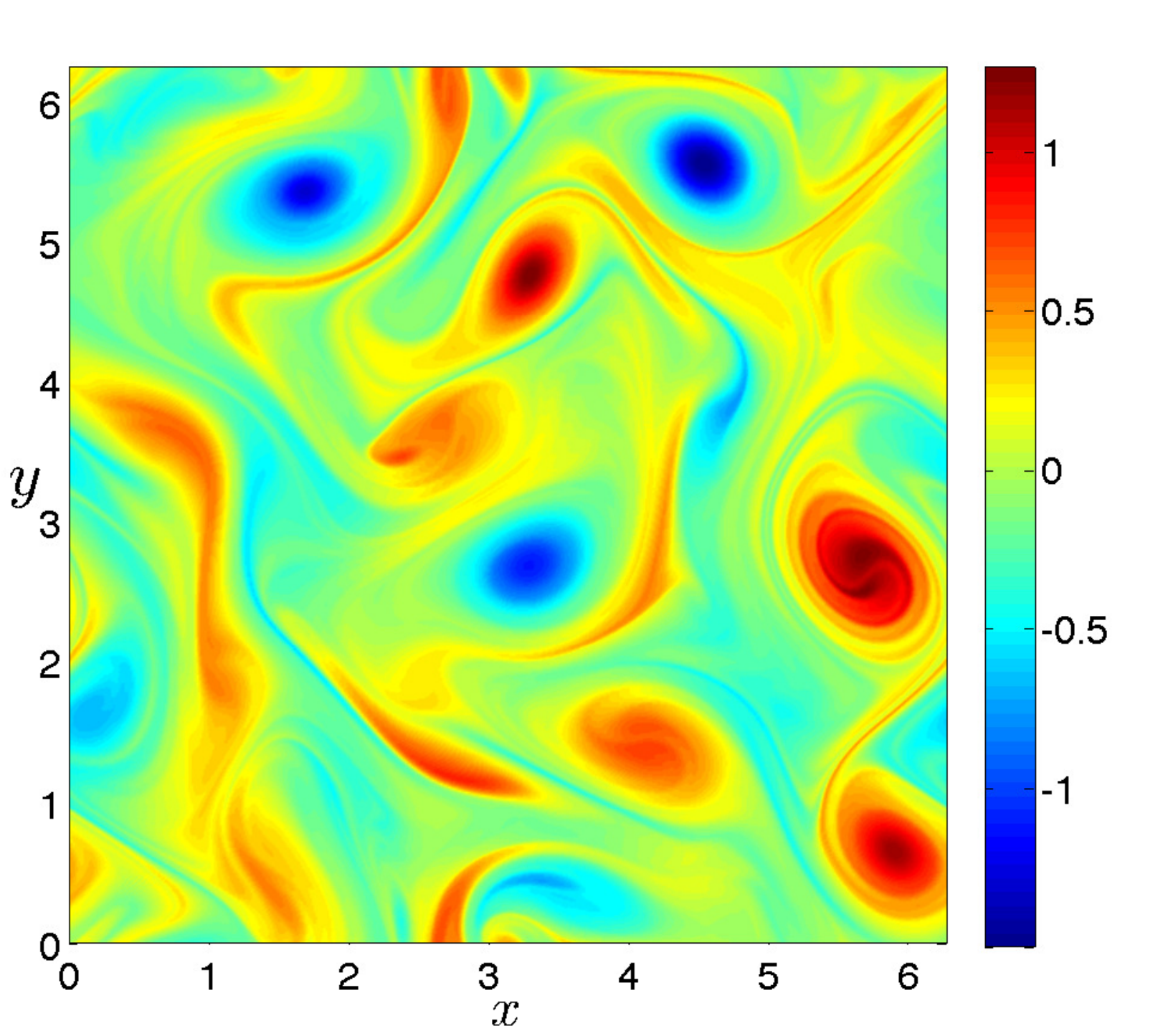}
	\hspace{.02\textwidth}
	\includegraphics[width=0.46\textwidth]{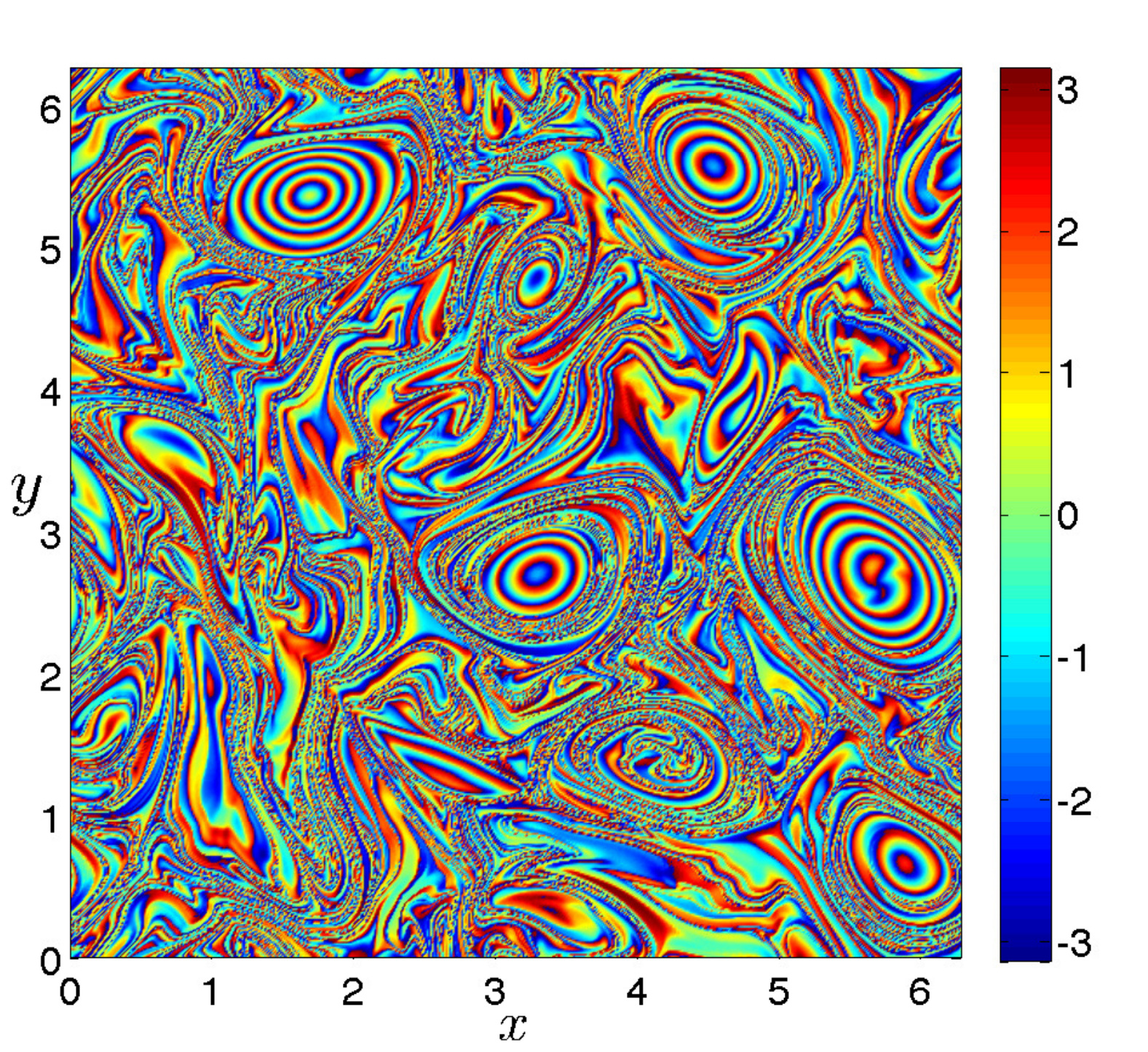} 
	\caption{Left: 
		Vorticity $\omega$ at the initial time $t=50$. Right: The PRA $\theta_{t_0}^t$ 
		computed
		from formula \eqref{eq:fullpolar} over the time interval $[50,100]$.}
	\label{fig:2Dturb_LagRot}
\end{figure}

To illustrate the use of PRA in detecting elliptic islands in a turbulent
flow, we solve the Navier--Stokes equation \eqref{eq:nse} with $\nu=10^{-5}$
on the domain $\mathcal D=[0,2\pi]\times[0,2\pi]$ with periodic boundary conditions.
We use a pseudo-spectral method with $512^{2}$ modes to evaluate
the spatial partial derivatives and the nonlinear term. The external forcing
is random in phase and only active over the wave-numbers $3.5<k<4.5$.
The forcing amplitude is time-dependent and chosen to balance the
instantaneous enstrophy dissipation $-\nu\int_{U}|\bnabla\omega(\vc x,t)|^{2}\id\vc x$.
The time integration is carried out by a variable step-size, fourth-order
Runge--Kutta method \citep{ode45}. We solve the equation up to time
$t=100$. We observe that the turbulent flow is fully developed after
$50$ time units. Therefore, we choose times $t_0=50$ and $t=100$ as the initial and 
final times for the computation of the PRA $\theta_{t_0}^t$.

Such two-dimensional turbulent flows tend to generate long-lasting
coherent vortices \citep{mcwilliams1990vortices}, which are also
prevalent in geophysical flows \citep{GFDfund}. Highly coherent Lagrangian
signatures of such vortices have been recently identified as regions
bounded by uniformly stretching material lines \citep{bhEddy,LCS_review}.

Here, we take an alternative approach and identify coherent Lagrangian
vortices as regions filled with polar LCSs. In other words, we seek
the elliptic islands of turbulence as regions of closed
material lines that pointwise have the same rigid-body rotation component
in their deformation over the time interval of interest.

Figure~\ref{fig:2Dturb_LagRot} (right panel) shows the PRA computed from formula
\eqref{eq:fullpolar} for $512\times512$ uniformly distributed initial
conditions. The polar LCSs are clearly visible as concentric closed
contours of $\theta_{t_{0}}^{t}(\vc x_{0})$. Figure~\ref{fig:2Dturb_contour}
shows a closeup view of a coherent Lagrangian vortex identified from
the PRA plot. Note how the PRA shows a sharp distinction between the
vortical region and the surrounding chaotic background. As in the
case of the standard map (see Fig.~\ref{fig:standardMap}), the chaotic
region is marked by small-scale, sharp variations of the Lagrangian
rotation due to sensitive dependence of material rotation angle on
initial conditions. 
\begin{figure}[h!]
	\centering \includegraphics[width=1\textwidth]{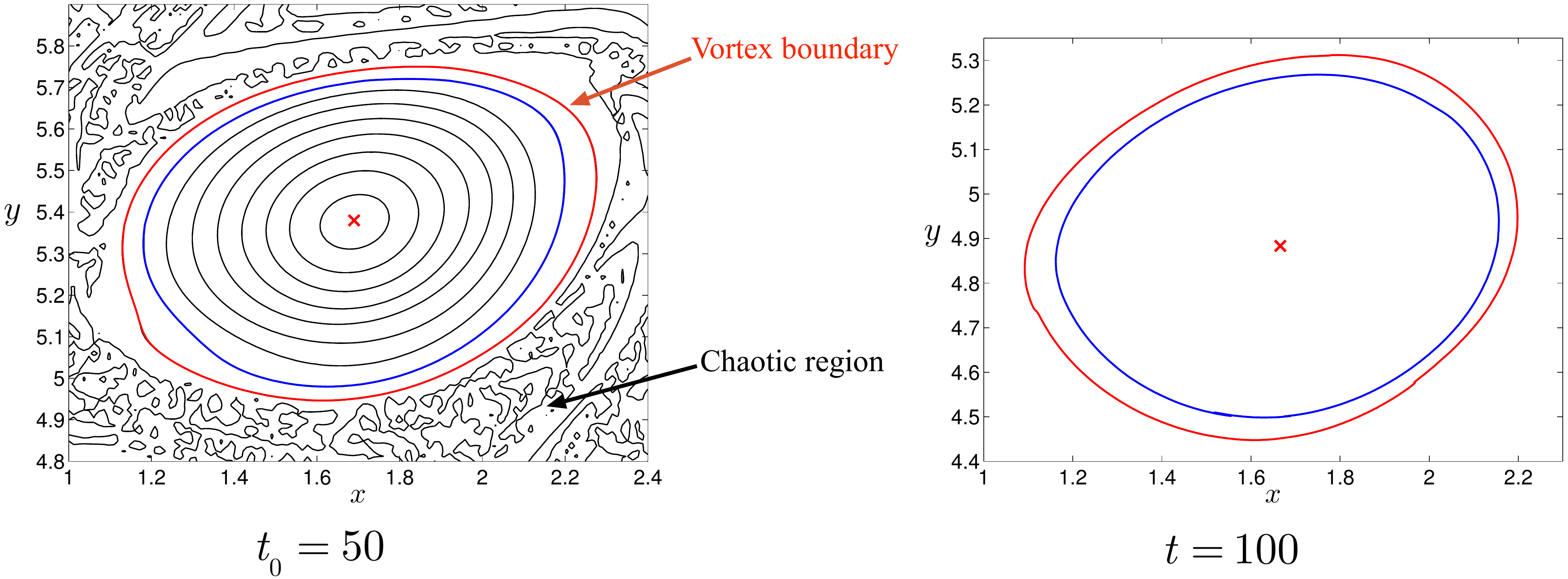}
	\protect\caption{Left: Contours of the PRA signaling coherent and chaotic regions.
		Right: Advected image of select contours to the final time $t=100$.
		The local extremum of the PRA (marked by a cross) defining the Lagragian
		vortex center by Definition 3.}
	\label{fig:2Dturb_contour} 
\end{figure}

While the velocity field is well-resolved, resolving small-scale Lagrangian
structures requires significantly higher resolution \citep{batchelor1959,aref,Berti2014}.
At the present resolution, the Lagrangian structures in the chaotic
region are not well-resolved. Nonetheless, the boundary of the vortex
can be approximated by the contour across which the PRA transitions
from concentric large-scale contours to small-scale sharp variations
(see the red-colored contour in Fig.~\ref{fig:2Dturb_contour}).

We now illustrate that the large-scale polar LCSs, defined by closed
contours of the PRA (cf. Definition \ref{def:polarLCS}) indeed remain coherent under
advection. We advect two such contours under the flow, with their
advected positions shown in the right panel on Fig.~\ref{fig:2Dturb_contour}
at time $t=100$. 

As a measure of coherence we define relative stretching of material lines as 
$\left[\ell(t)-\ell(t_0)\right]/\ell(t_0)$, where $\ell$ denotes the length of the 
material line as a function of time.
The relative stretching of the blue and red contours
are $2.65\%$ and $-1.38\%$, respectively. These relative stretching
values remain in the order of stretching exhibited by perfectly coherent
elliptic LCSs obtained from the geodesic LCS theory~\citep{bhEddy}.

\begin{figure}[h!]
	\centering 
	\includegraphics[width=1\textwidth]{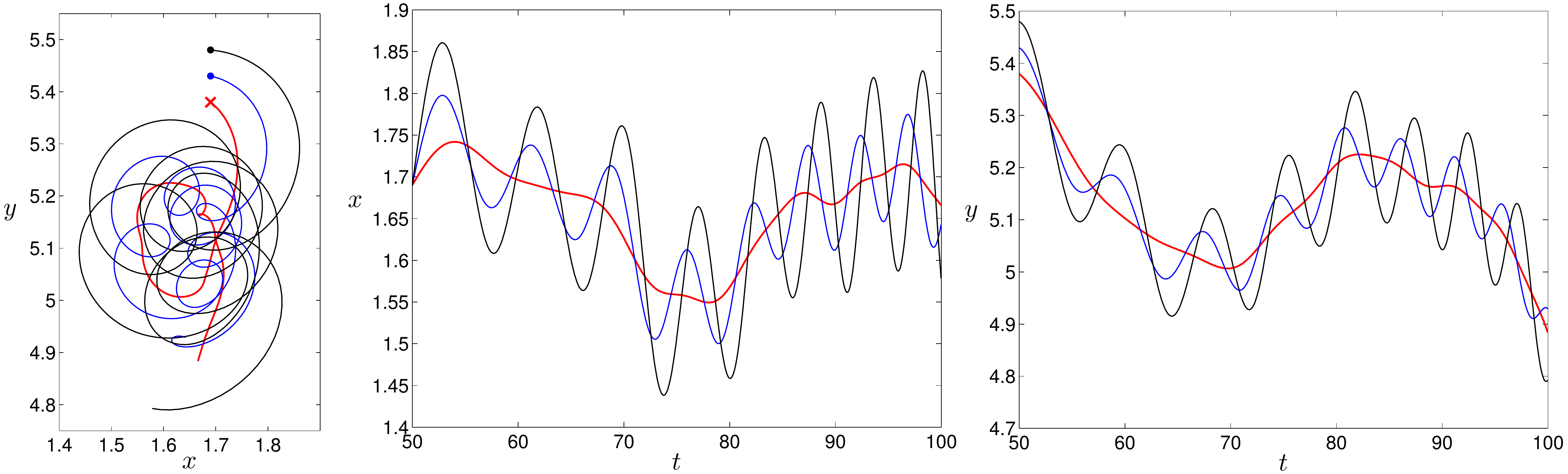}
	\protect\caption{Left panel: Trajectories of the Lagrangian vortex center (red) and
		nearby passive tracers (blue and black). Middle and Right panels:
		The coordinates of the vortex center and nearby tracers as a function
		of time.}
	\label{fig:2Dturb_vCore} 
\end{figure}

The cross in Fig.~\ref{fig:2Dturb_contour} marks a local extremum
of the PRA, which is a Lagrangian vortex center by our Definition
\ref{def:vc}. This local extremum indeed turns out to behave as the vortex center
over the time interval of interest, i.e., $t\in[50,100]$. 
Figure~\ref{fig:2Dturb_vCore} 
shows the trajectory starting from this PRA
extremum, whose initial coordinates are $(1.690,5.380)$. For reference,
two other trajectories are also shown with initial positions at
$0.05$ and $0.1$ distance from the vortex center. Due to the complexity
of the flow, the trajectory patterns are not illuminating. However,
their $x$- and $y$-coordinates as a function of time reveal the
oscillatory motion of nearby trajectories around the vortex center,
while the vortex center itself has minimal oscillations (cf. middle
and right panels of Fig.~\ref{fig:2Dturb_vCore}). The oscillations
of the vortex center are due to the motion of the vortex as a whole.
The nearby trajectories, however, exhibit higher frequency oscillations
which are due to their swirling motion around the vortex center.

\subsection{Stratified geophysical fluid flow}\label{sec:BVE} 
We consider a simplified model for stratified geophysical
fluid flow, the barotropic equation. This equation, in the vorticity-stream
form, reads \citep{MajdaBook} 
\begin{equation}
\partial_{t}\omega+J(\psi,\omega)+\partial_{x}\psi=0,\ \ \ 
\omega=\Delta\psi,\label{eq:bve}
\end{equation}
where $w(x,y,t)$ and $\psi(x,y,t)$ are non-dimensional vorticity
and stream function, respectively. In deriving this equation, the
viscous dissipation is neglected and the Coriolis frequency is assumed
to be linear in the meridional coordinate $y$ (i.e., the $\beta$-plane
approximation is used \citep{MajdaBook}). The Jacobian operator reads
$J(\psi,\omega)=\partial_{x}\psi\partial_{y}\omega-\partial_{y}\psi\partial_{x}\omega$.
The fluid velocity field $\vc u$ is given in terms of the stream
function by $\vc u=(\partial_{y}\psi,-\partial_{x}\psi)$.

\begin{figure}[t!]
	\centering 
	\includegraphics[width=\textwidth]{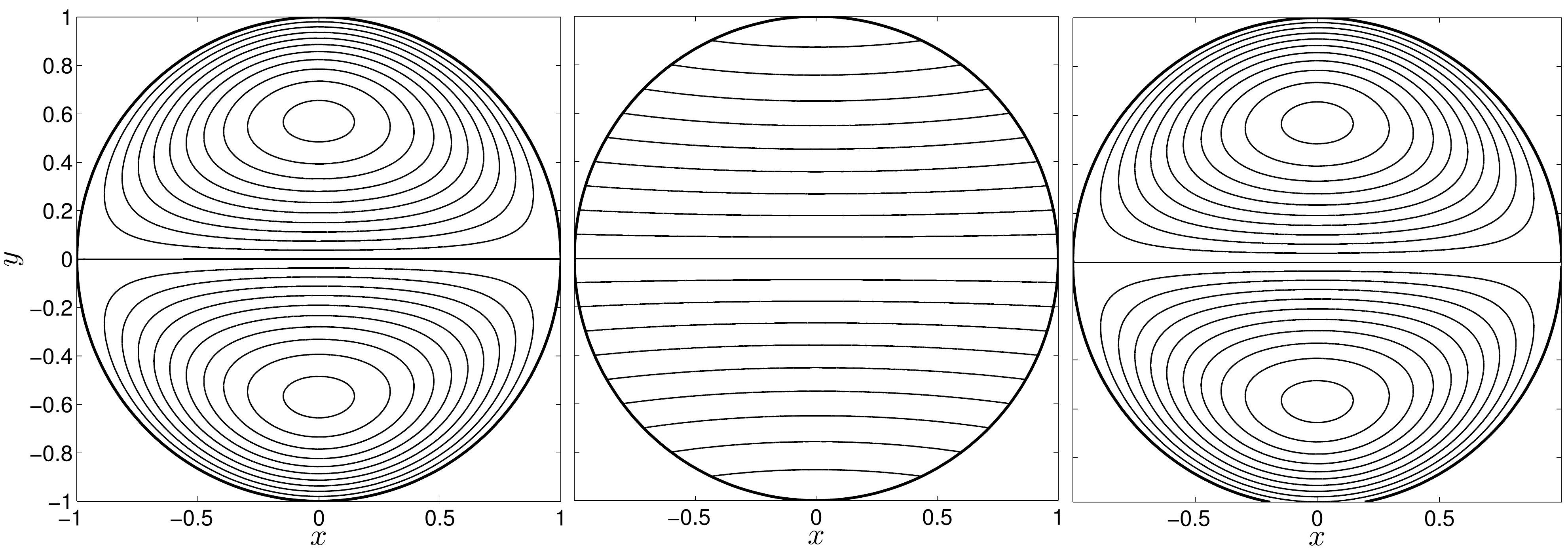}
	\caption{Contours of the stream function $\psi$ (left), vorticity $\omega$
		(middle) and potential vorticity $q$ (right) of the modon solution
		\eqref{eq:modon}.}
	\label{fig:modons} 
\end{figure}

\begin{figure}[h!]
	\centering 
	\includegraphics[width=\textwidth]{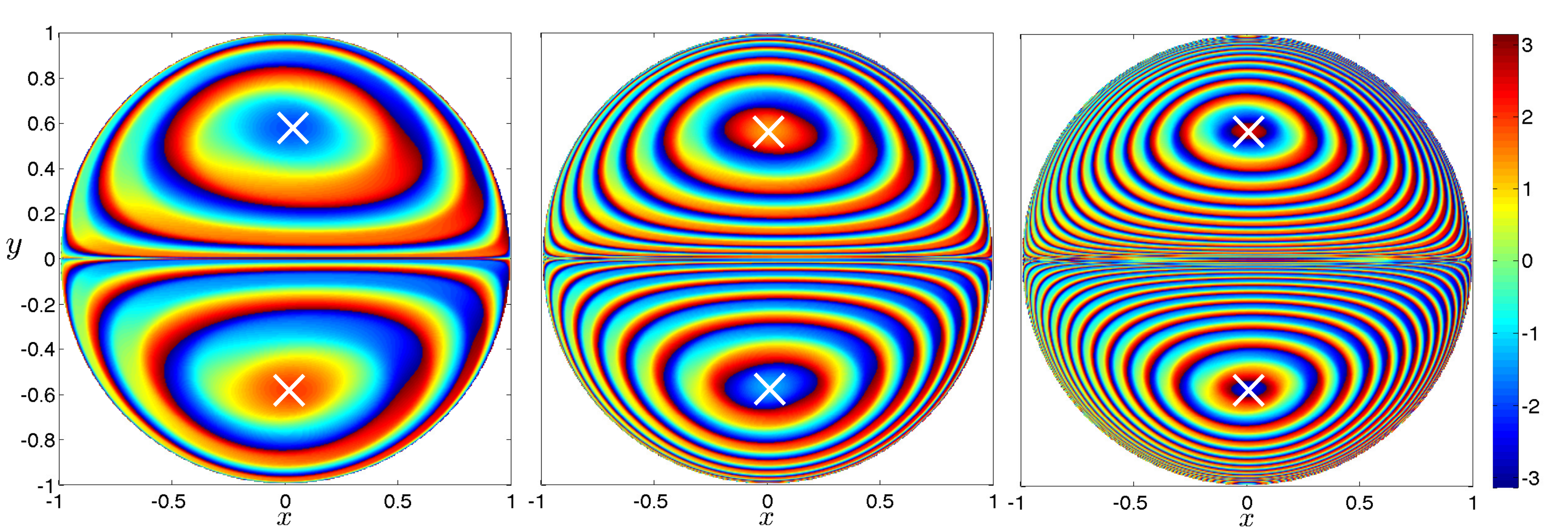}
	\protect\caption{The PRA for the modon solution \eqref{eq:modon}. The initial time
		is $t_{0}=0$ and the final times are $t=100$ (left), $t=250$ (middle)
		and $t=500$ (right). All figures are computed from a grid of roughly
		$35$ thousand uniformly distributed initial conditions in the unit
		disk. The points marked by crosses are the Lagrangian vortex centers obtained 
		from 
		Definition~\ref{def:vc} over the corresponding time intervals.}
	\label{fig:modons_LagRot} 
\end{figure}

Vorticity is not preserved along fluid trajectories when the flow
satisfies \eqref{eq:bve}. Instead, one can show that the \emph{potential
	vorticity} $q=\omega+y$ is conserved along these trajectories (see, e.g., 
\citep{MajdaBook}).

We consider a steady exact solution of the barotropic equation \eqref{eq:bve}
called a \emph{modon}: a uniformly propagating vortex dipole. For
this modon solution, the stream function and vorticity are given respectively
by 
\begin{align}
& \psi(r,\varphi)=\left(\frac{J_{1}(r)}{J_{1}(1)}-r\right)\sin\varphi, & 0\leq 
r\leq1\nonumber \\
& \omega(r,\varphi)=-\frac{J_{1}(r)}{J_{1}(1)}\sin\varphi, & 0\leq 
r\leq1\label{eq:modon}
\end{align}
in polar coordinates $(r,\varphi)$ where $r=\sqrt{x^{2}+y^{2}}$,
$\tan\varphi=y/x$ and $J_{1}$ is the Bessel function of the first
kind \citep{Flierl1980}. This solution is written in a frame co-moving
with the modon at a constant speed $c=1$.

The stream function $\psi$ defines a flow on the invariant domain
$r\leq1$. While this solution can, in principle, be extended to the
entire plane \citep{Flierl1980,BarEqExSol_03}, here we only consider
the motion inside the unit disk.

Figure~\ref{fig:modons} shows the stream function, the vorticity and
the potential vorticity for the modon solution \eqref{eq:modon}.
Since the flow is integrable, its stream function completely describes
the flow structure, showing two counter-rotating vortices.

The vorticity $\omega$ is negative in the upper half-disk $y>0$
and positive in the lower half-disk $y<0$. Its contours, however,
do not reveal the two vortices present in the flow. This is because
unlike the two-dimensional Euler flows, vorticity is not conserved
along the trajectories of the solutions of the barotropic equation
\eqref{eq:bve}.

The potential vorticity $q$, as a conserved quantity, reveals the
eddies. Its level curves (Fig.~\ref{fig:modons}, right panel) resemble
those of the streamlines. In fact, the particular solution \eqref{eq:modon} of the 
barotropic equation satisfies $q=-\psi$.

Figure~\ref{fig:modons_LagRot} shows the PRA for integration times
$t=100$, $250$ and $500$. The integration time $t=100$ is chosen
such that almost all periodic orbits of $\dot{\vc x}=\vc u(\vc x)$
complete at least one period. Even with this relatively short integration
time, PRA contours already reveal the vortices. Obtained from a finite-time
assessment of the flow, the PRA contours deviate from the trajectories.
As the integration time increases, however, the PRA contours converge
to the streamlines and Lagrangian vortex centers obtained from the
PRA converge to the elliptic fixed points of the flow.

\subsection{ABC flow}
\label{sec:abc} 
As our last example, we consider the Arnold-Beltrami--Childress (ABC)
flow $\dot{\vc x}=\vc u(\vc x)$ where 
\begin{equation}
\vc u(\vc x)=\begin{pmatrix}A\sin(z)+C\cos(y)\\
B\sin(x)+A\cos(z)\\
C\sin(y)+B\cos(x)
\end{pmatrix},\label{eq:abc}
\end{equation}
with $\vc x=(x,y,z)$ and $A,B,C\in\mathbb{R}$ are constant parameters
\citep{topolHydro_arnold}. The velocity field $\vc u$ is an exact
steady solution of Euler's equation for inviscid Newtonian fluids
with periodic boundary conditions. The ABC velocity field is a Beltrami
vector field satisfying $\pmb\omega(\vc x)=\vc u(\vc x)$ with 
$\pmb\omega=\bnabla\times\vc u$
being the vorticity field. 

In the following, we set $A=1$, $B=\sqrt{2/3}$ and 
$C=\sqrt{1/3}$. The Lagrangian computations are carried out on a uniform
grid of $200\times 200\times 200$ initial conditions distributed over the domain $\mathbb 
T^3\in[0,2\pi]\times[0,2\pi]\times[0,2\pi]$.

Figure~\ref{fig:abc_vort_LagRot} (left panel) shows the helicity density
$\langle\vc u,\pmb\omega\rangle$ (=$|\pmb\omega|^{2}$). While such Eulerian features
may suggest coherent vortical motion throughout the domain, the ABC
flow is known to have chaotic fluid trajectories in addition to coherent
swirling trajectories lying on invariant tori \citep{abc_chaos}.
\begin{figure}[t!]
	\centering \includegraphics[width=0.48\textwidth]{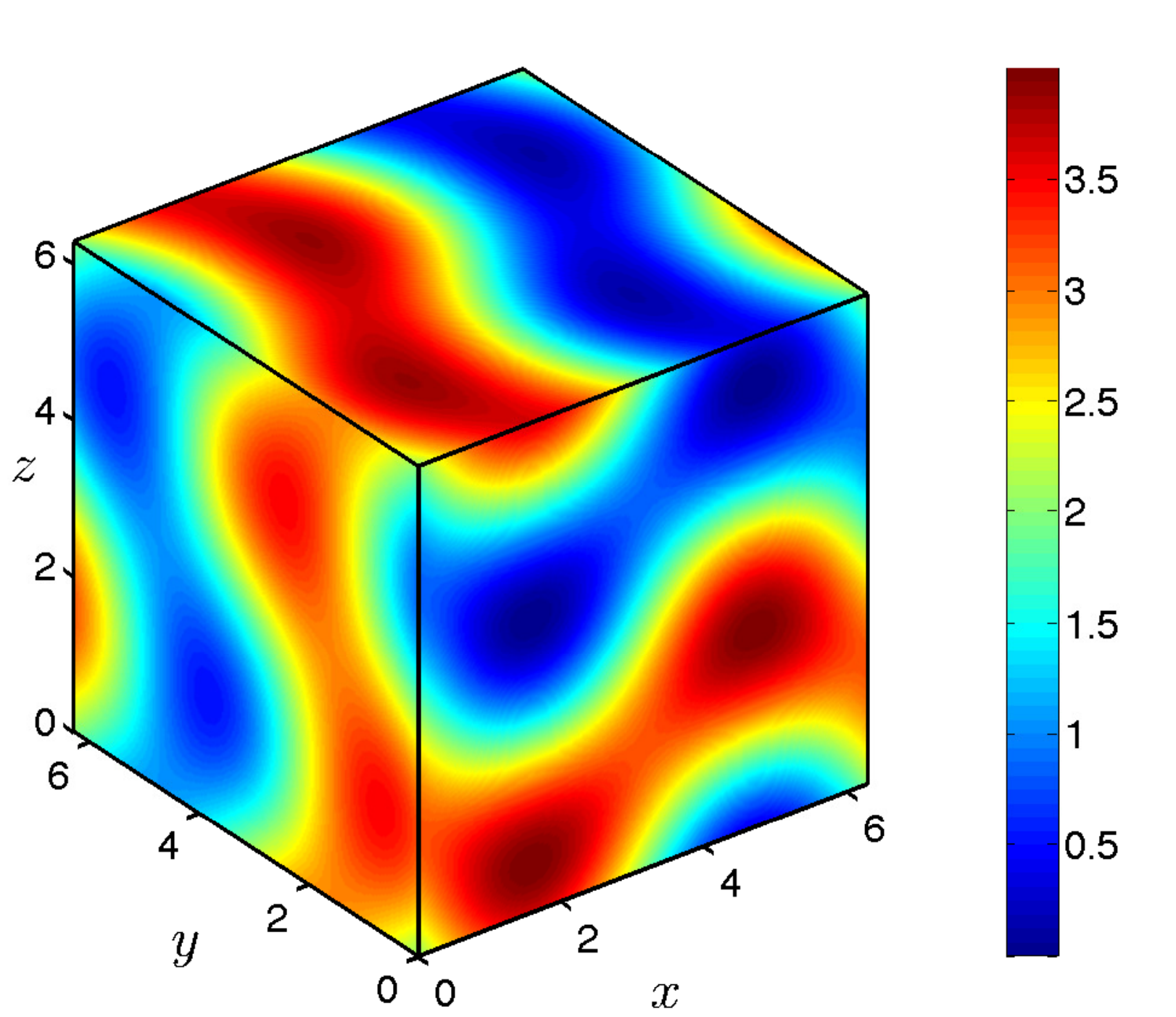}
	\includegraphics[width=0.48\textwidth]{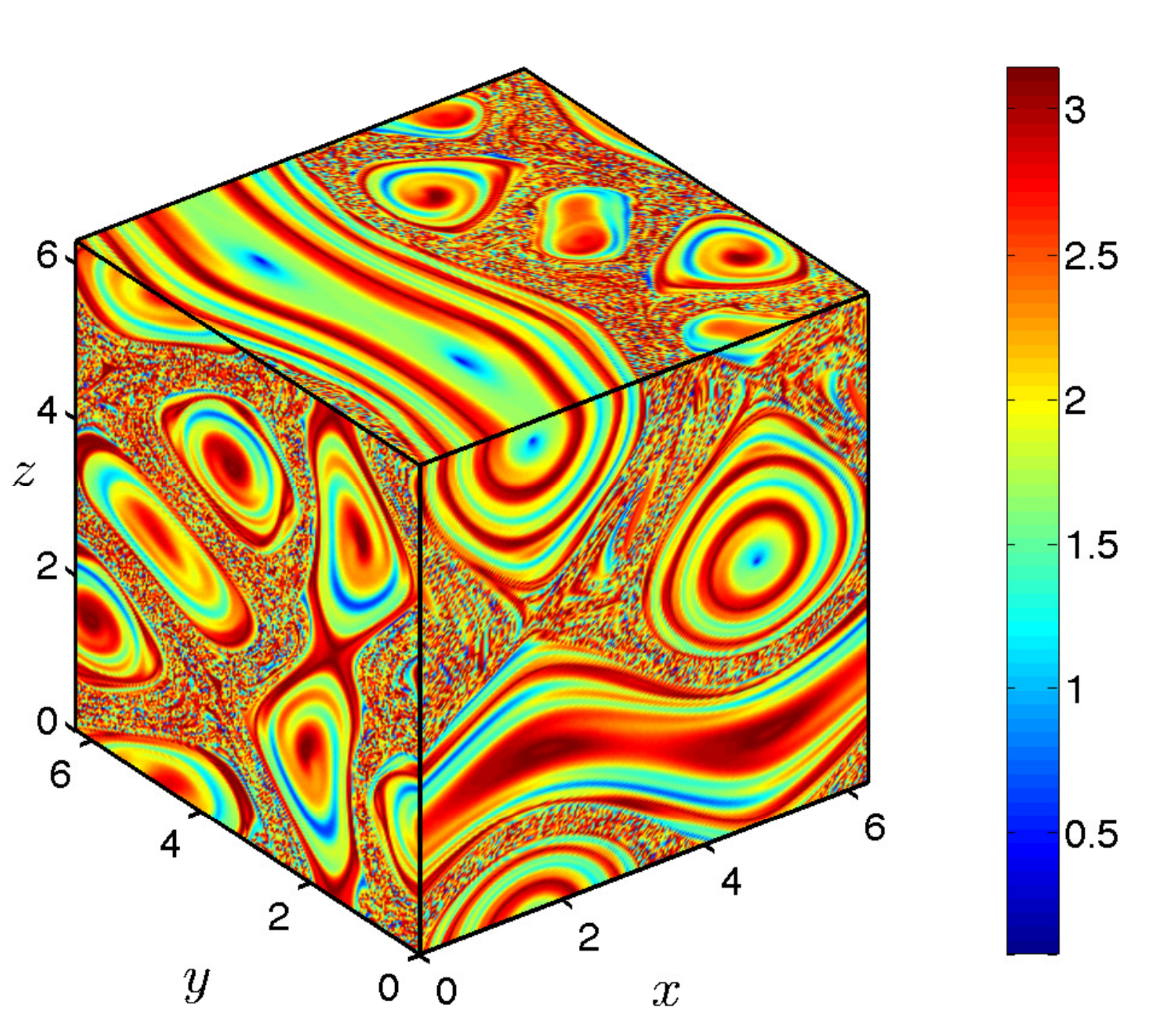}
	\protect\caption{Left: The helicity $\langle\vc u,\pmb\omega\rangle=|\pmb\omega|^{2}$
		for the ABC flow. Right: The two-quadrant PRA $\theta_{t_{0}}^{t}$ with the 
		integration
		time $t-t_0=50$, computed from formula \eqref{eq:th_2q}.}
	\label{fig:abc_vort_LagRot} 
\end{figure}
\begin{figure}[t!]
	\centering 
	\includegraphics[width=0.45\textwidth]{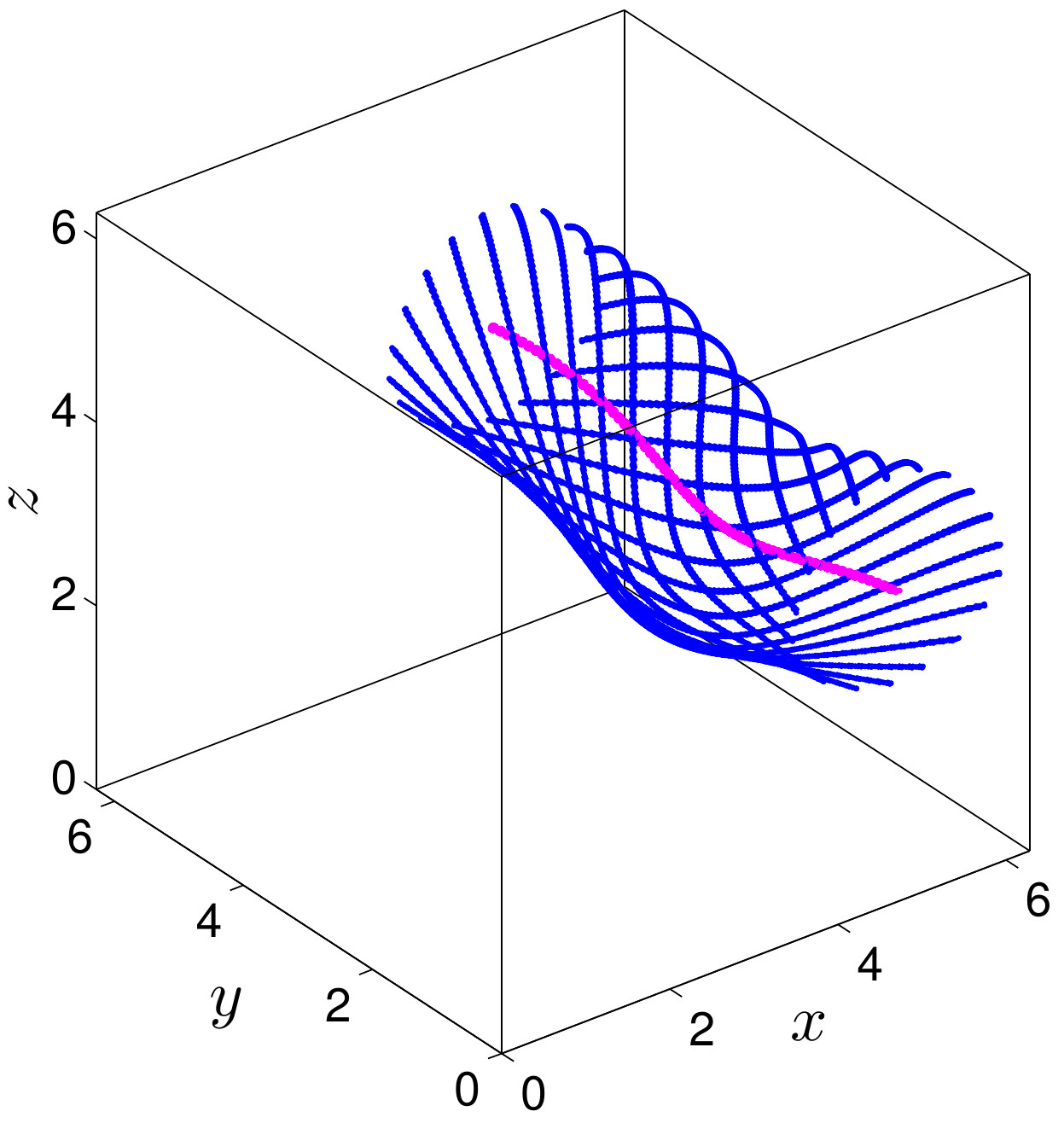}
	\includegraphics[width=0.45\textwidth]{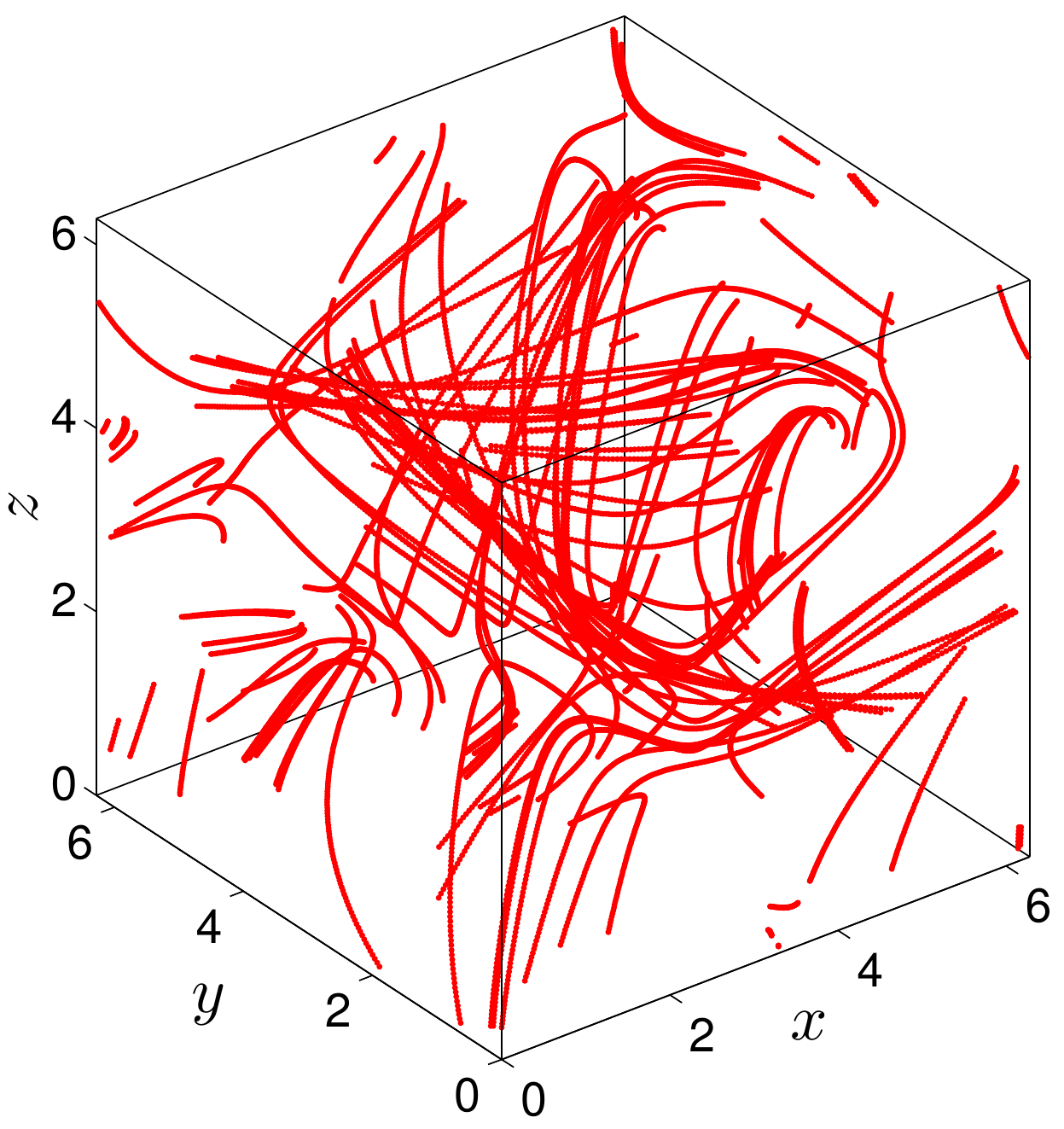}
	\includegraphics[width=0.35\textwidth]{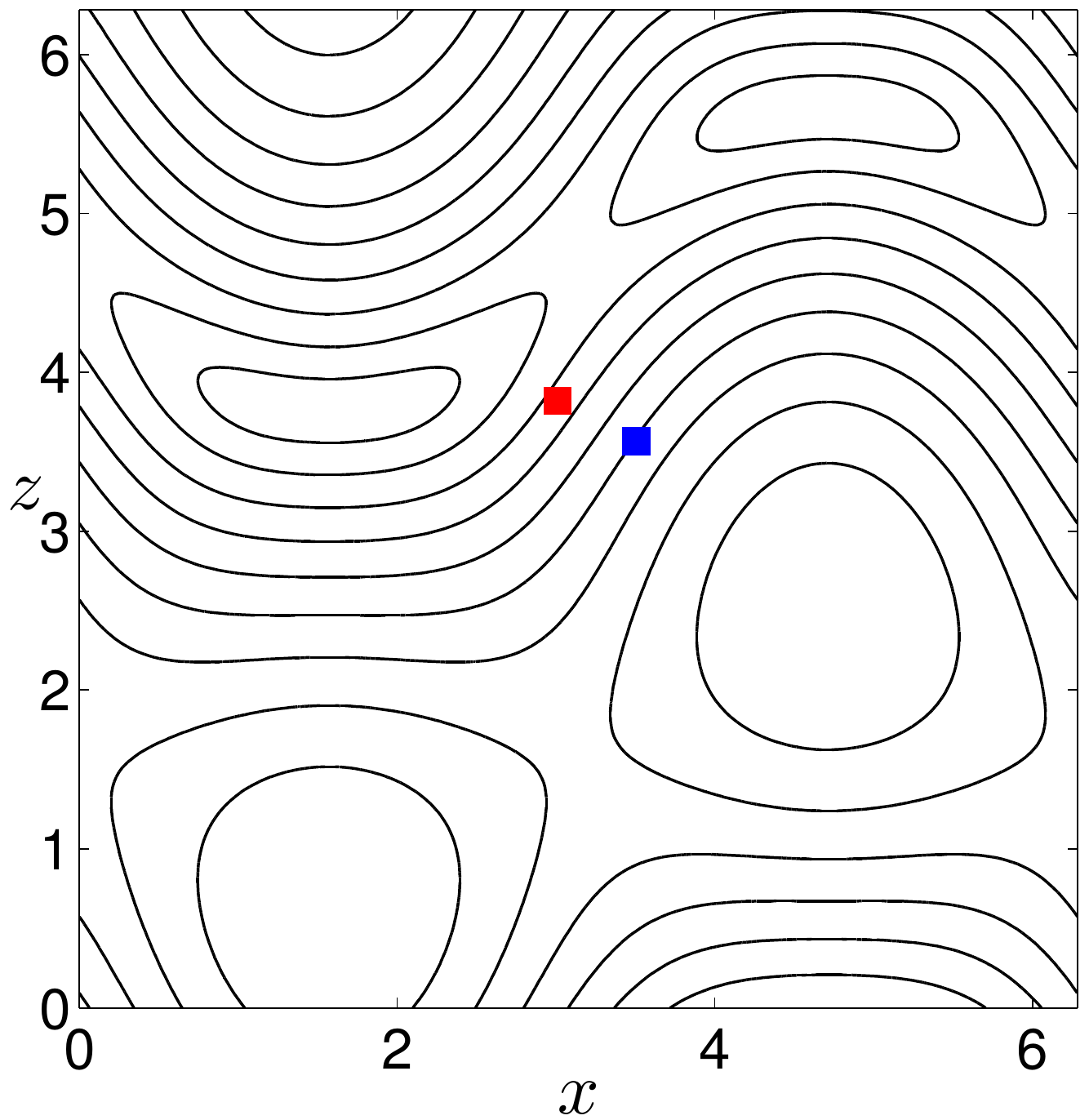}
	\includegraphics[width=0.35\textwidth]{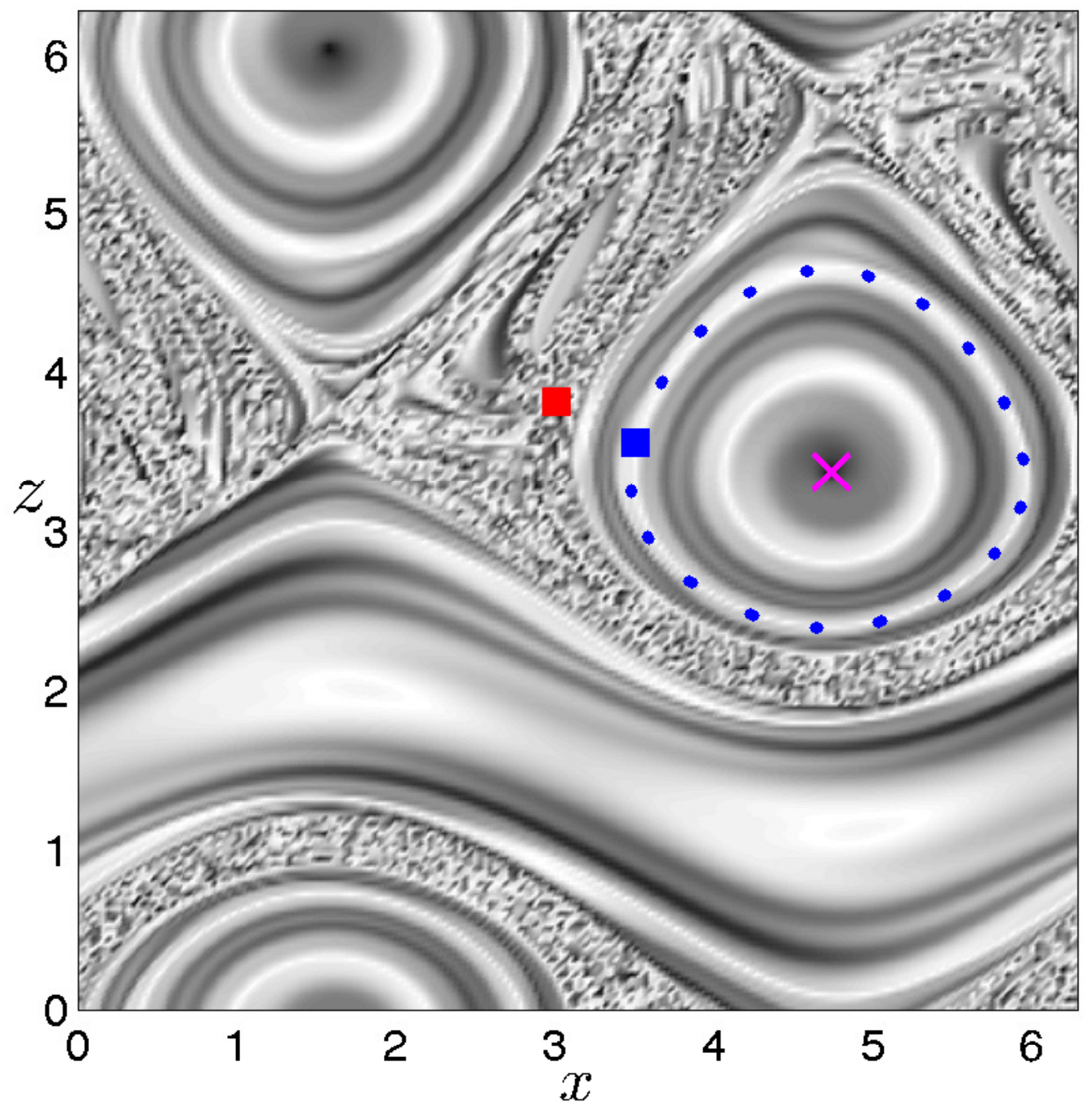}
	\protect\caption{Top: Two trajectories of the ABC flow. The blue trajectory 
		starts
		in an elliptic island and traces the surface of an invariant torus.
		The red trajectory is chaotic. Note that the trajectories are plotted
		modulo $2\pi$. The material line marking the Lagrangian vortex core (cf. 
		Definition 
		\ref{def:vc}) is 
		plotted in magenta color. 
		Bottom: The initial condition (squares) of each
		trajectory is superimposed on the $y=0$ slice of the helicity $\langle\vc 
		u,\pmb\omega\rangle=|\pmb\omega|^{2}$
		(left) and the PRA $\theta_{t_{0}}^{t}$ obtained from formula
		\eqref{eq:th_2q} (right). The right panel also shows the intersections
		of the blue trajectory with the plane $y=0$ (blue dots), as well as the vortex 
		core (magenta cross).}
	\label{fig:abc_trajs} 
\end{figure}

These invariant tori form vortical regions that we seek to capture
from finite-time flow samples as elliptic regions. Using a local variational
principle extremizing Lagrangian shear, elliptic LCSs approximating
the tori from finite-time flow samples have been constructed by \citet{blazevski_3d}.
Here we illustrate that polar LCSs obtained from the PRA also give
a close approximation at a reduced computational cost.

Indeed, the PRA admits tubular level surfaces that closely approximate
the invariant tori (Fig.~\ref{fig:abc_vort_LagRot}, right panel).
Codimension-two level sets of the PRA are periodic material curves
at the cores of the elliptic regions. These material lines serve as
Lagrangian vortex centers by Definition~\ref{def:vc}. As in earlier examples, outside
the elliptic islands formed by these closed level surfaces, PRA levels
exhibit small-scale variations due to sensitive dependence of the
rotation angle on the initial conditions.

To examine how accurately the PRA field $\theta_{t_{0}}^{t}$ captures
the tori and the chaotic region boundaries, we release two trajectories
from the initial conditions $\vc x_{0}=(3.085,0,3.820)$ (red square
in the bottom panel of Fig.~\ref{fig:abc_trajs}) and $\vc x_{0}=(3.505,0,3.568)$
(blue square in the bottom panel of Fig.~\ref{fig:abc_trajs}). The
initial conditions are chosen such that they are nearby, yet one belongs
to the chaotic region (red square) and the other (blue square) belongs
to a smooth level-surface of the PRA signaling an invariant torus.

These initial conditions are then advected under the ABC flow from
time $t=0$ to $t=500$. The resulting trajectories are shown
in the top panel of Fig.~\ref{fig:abc_trajs}. As expected, the blue
trajectory remains on a torus while the red trajectory exhibits chaotic
behavior. Note that all curves correspond
to a single trajectory and only appear as line segments because
they are plotted modulo $2\pi$. The intersections of the coherent trajectory with the 
Poincar\'e section $y=0$ shows that the PRA captures the invariant torus accurately
(Fig.~\ref{fig:abc_trajs}, bottom panel).

We stress that both initial conditions studied here belong to topologically
equivalent regions of the local helicity $\langle\vc u,
\pmb\omega\rangle=|\pmb\omega|^{2}$. The
vorticity magnitude, therefore, fails to distinguish vortical regions
from chaotic regions. This is because vorticity magnitude
is not a material invariant of the Euler's equation in three dimensions
and therefore does not generally capture material behavior.

\section{Conclusions}

Most approaches to coherent structures seek their
signature in material separation or stretching. By contrast, we have
developed here an approach to locate coherent structures based
on their signature in material rotation. To quantify finite material
rotation in a mathematically precise fashion, we have used the polar
rotation tensor from the unique rotation-stretch factorization of
the deformation gradient.

For two- and three-dimensional dynamical systems, we have derived
explicit formulas for the polar rotation angle (PRA) generated by
the rotation tensor around its axis of rotation. While polar rotation
has broadly been studied and used in continuum mechanics, the simple
formulas we have derived here for the PRA in terms of the flow gradient,
its singular values and singular vectors have not been available.
These formulas enable the efficient computation of the PRA from basic
quantities provided by existing numerical algorithms for Lagrangian
coherent structure detection.

Building on the PRA, we have also introduced the notion of polar Lagrangian
coherent structures (polar LCSs). These are tubular material surfaces
along which trajectories admit the same PRA value over a finite time
interval of interest. We have proposed regions filled by polar LCSs
as rotation-based generalizations of the classic elliptic islands
filled by KAM tori in Hamiltonian systems.

As we demonstrated on a direct numerical simulation of two-dimensional
turbulence, the PRA identifies Lagrangian vortex boundaries with high accuracy.
While geodesic LCS theory of \citet{bhEddy} offers an exact detection of 
such vortex boundaries as solutions of differential equations, the present
diagnostic detection of these boundaries as outermost closed PRA level
curves is substantially less computational, and hence preferable for an approximate
identification of these boundaries.

Outside the Lagrangian vortex boundaries, the PRA is dominated by
small-scale noise due to its sensitive dependence on initial conditions.
In these regions, therefore, the PRA displays no clear signature for
hyperbolic LCSs governing chaotic tracer mixing. These latter types
of LCSs, by contrast, are efficiently revealed by another objective
diagnostic, the finite-time Lyapunov exponent (FTLE) \citep{LCS_review}.
The PRA and FTLE have a well-defined duality: the former is a scalar
field characterizing the rotational factor $\mathbf{R}_{t_{0}}^{t}$,
while the latter characterizes the stretch factor $\mathbf{U}_{t_{0}}^{t}$
in the polar decomposition $\bnabla\Ff=\R\mathbf U_{t_0}^t$ of the deformation gradient.

We have found that local extrema of the PRA mark initial positions
of trajectories that serve as well-defined centers for elliptic islands. 
Oscillations in these center trajectories are minimal and
arise solely due to the material translation of the underlying island.
Nearby trajectories inside the elliptic island, on the other hand, oscillate rapidly
due to their swirling motion around the center trajectory (cf. 
Fig.~\ref{fig:2Dturb_vCore}). The ability of the PRA to identify a unique
vortex center should be helpful in Lagrangian versions of the Eulerian
eddy censuses carried out by \citet{dong2014global} and \citet{chelton2007global}.

The elliptic island boundaries marked by PRA do not necessarily remain unfilamented
under advection. If the goal is to find perfectly coherent Lagrangian
vortices (see, e.g.,~\citep{farazmand2014cohVort}), then the geodesic
theory of \citet{bhEddy} should be applied. This theory identifies
material vortex boundaries as closed null geodesics of the generalized
Green-Lagrange strain tensor. The related computations require the
a priori identification of phase space regions where such closed geodesics
may exist \citep{karrasch2014automated}. Vortex regions identified
from the PRA provide a quickly computable starting point for the detection
of closed Green-Lagrange null geodesics. Incorporating the vortex
centers obtained from the PRA in the geodesic LCS analysis
is, therefore, expected to lead to a notable computational speed-up.

Finally, the polar LCSs obtained as level curves of the PRA are frame-invariant
for planar flows (see Proposition \ref{prop:obj}). Such objectivity is desirable
for coherent structure identification methods in order to exclude
false positives and negatives specific to the coordinate system used
in the analysis \citep{LCS_review}. In three dimensions, however,
the PRA does depend on the reference frame. The objective detection
of higher-dimensional elliptic islands from their rotational coherence,
therefore, requires further work.

\begin{appendices}
	\section{Proof of Proposition \ref{prop:pra}}\label{app:proof_prop1}
	\textbf{Part (1):} The trace of a tensor is independent of the choice of basis. If we 
	represent the rotation 
	tensor $\R$ in the orthonormal basis
	$\{\bxi_{k}\}_{1\leq k\leq 3}$, then its entires satisfy 
	$\left[\R\right]_{ij}=\left\langle \bxi_{i},\R\bxi_{j}\right\rangle $.
	Therefore, using formula \eqref{eq:rescaling}, we can write
	\begin{align}
	\tr\R & =\sum_{i=1}^{3}\left\langle \bxi_{i},\R\bxi_{i}\right\rangle 
	=\sum_{i=1}^{3}\left\langle 
	\bxi_{i},\nabla\Ff\left[\mathbf{U}_{t_{0}}^{t}\right]^{-1}\bxi_{i}\right\rangle
	\nonumber\\
	& =\sum_{i=1}^{3}\left\langle 
	\bxi_{i},\nabla\Ff\frac{1}{\sqrt{\lambda_{i}}}\bxi_{i}\right\rangle 
	=\sum_{i=1}^{3}\frac{\left\langle \bxi_{i},\nabla\Ff\bxi_{i}\right\rangle 
	}{\sqrt{\lambda_{i}}},
	\label{eq:traceform}
	\end{align}
	which, together with \eqref{eq:LagRot_3d_cos}, proves formula \eqref{eq:cosPRA_3d}.
	
	To prove formula \eqref{eq:sinPRA_3d}, we first note the
	coordinate form of equation \eqref{eq:LagRot_3d_sin}:
	\[
	\frac{1}{2}\left(\left[\mathbf{R}_{t_{0}}^{t}\right]_{ij}-\left[\mathbf{R}_{t_{0}}^{t}\right]_{ji}\right)=\sin\theta_{t_{0}}^{t}\epsilon_{ijk}\left[\mathbf{r}_{t_{0}}^{t}\right]_{k}.
	\]
	Applying the same argument used in \eqref{eq:traceform} in the strain
	eigenbasis, we obtain
	\begin{equation}
	\sin\theta_{t_{0}}^{t}=\frac{\left\langle 
		\bxi_{i},\nabla\Ff\frac{1}{\sqrt{\lambda_{j}}}\bxi_{j}\right\rangle -\left\langle 
		\bxi_{j},\nabla\Ff\frac{1}{\sqrt{\lambda_{i}}}\bxi_{i}\right\rangle 
	}{2\epsilon_{ijk}\left[\mathbf{r}_{t_{0}}^{t}\right]_{k}},\quad i\neq 
	j.\label{eq:sin_proof}
	\end{equation}
	Next we write the eigenvector $\mathbf{r}_{t_{0}}^{t}$ in strain
	basis as $\mathbf{r}_{t_{0}}^{t}=\sum_{k}e_{k}\bxi_{k}$ to obtain
	\begin{equation*}
	\sum_{k}e_{k}\bxi_{k}=\mathbf{R}_{t_{0}}^{t}\sum_{k}e_{k}\bxi_{k}=
	\nabla\Ff\left[\mathbf{U}_{t_{0}}^{t}\right]^{-1}\sum_{k}e_{k}\bxi_{k}=
	\sum_{k}\frac{e_{k}}{\sqrt{\lambda_{k}}}\nabla\Ff\bxi_{k},
	\end{equation*}
	which implies 
	\[
	e_{j}=\sum_{k}\frac{\left\langle \bxi_{j},\nabla\Ff\bxi_{k}\right\rangle 
	}{\sqrt{\lambda_{k}}}e_{k},
	\]
	or, equivalently, $\mathbf{K}_{t_{0}}^{t}\mathbf{e}=\mathbf{e}$,
	with $\mathbf{K}_{t_{0}}^{t}$ and $\mathbf{e}$ defined in the statement
	of Proposition 1. Since $\left[\mathbf{r}_{t_{0}}^{t}\right]_{k}=e_{k}$,
	formula \eqref{eq:sin_proof} proves \eqref{eq:sinPRA_3d}.
	
	\textbf{Part (2)}: Two-dimensional flows are parallel to a distinguished plane and 
	exhibit
	no stretching or shrinking along the normal of this plane. In this
	case, we have
	\[
	\lambda_{1}\leq\lambda_{2}=1\leq\lambda_{3},
	\]
	with the strain eigenvector $\bxi_{2}$ pointing in the normal of
	the plane in question. Formula \eqref{eq:PRA3D} then gives 
	\begin{equation}
	\cos\theta_{t_{0}}^{t}=\frac{1}{2}\left(\frac{\left\langle 
	\bxi_{1},\nabla\Ff\bxi_{1}\right\rangle }{\sqrt{\lambda_{1}}}+\frac{\left\langle 
	\bxi_{3},\nabla\Ff\bxi_{3}\right\rangle 
	}{\sqrt{\lambda_{3}}}\right).\label{eq:2Dproof}
	\end{equation}
	Restricting our consideration to the two-dimensional plane of the
	flow, we reindex the quantities in formula \eqref{eq:2Dproof} as
	$\lambda_{3}\to\lambda_{2}$ and $\bxi_{3}\to\bxi_{2}$, given that
	the original $\lambda_{3}$ strain eigenvalue of the flow is the second
	largest principal strain in the plane of the flow. After this re-indexing,
	equation \eqref{eq:2Dproof} gives 
	\begin{equation}
	\cos\theta_{t_{0}}^{t}=\frac{1}{2}\sum_{i=1}^{2}\frac{\left\langle 
	\bxi_{i},\nabla\Ff\bxi_{i}\right\rangle }{\sqrt{\lambda_{i}}}.\label{eq:2Dprooflast}
	\end{equation}
	The summands in this last expression are just the diagonal elements
	of the two-dimensional rotation tensor $\R$ represented in the $\left\{ 
	\bxi_{1},\bxi_{2}\right\} $
	basis (cf. our discussion leading to equation \eqref{eq:traceform}). Since
	the diagonal elements of any two-dimensional rotation matrix are equal,
	formula~\eqref{eq:cosPRA} follows from \eqref{eq:2Dprooflast}. 
	
	In two dimensions, the rotation tensor is of the form 
	\begin{equation}
	\R(\vc x_{0})=\begin{pmatrix}\cos\theta_{t_{0}}^{t}(\vc x_{0}) & 
	\sin\theta_{t_{0}}^{t}(\vc x_{0})\\
	-\sin\theta_{t_{0}}^{t}(\vc x_{0}) & \cos\theta_{t_{0}}^{t}(\vc x_{0})
	\end{pmatrix}.
	\end{equation}
	Thus, using the argument in \eqref{eq:traceform}, we obtain that
	\[
	\sin\theta_{t_{0}}^{t}=\left\langle \bxi_{1},\R\bxi_{2}\right\rangle 
	=\frac{\left\langle 
		\bxi_{1},\nabla\Ff\bxi_{2}\right\rangle }{\sqrt{\lambda_{2}}}=-\left\langle 
	\bxi_{2},\R\bxi_{1}\right\rangle =-\frac{\left\langle 
		\bxi_{2},\nabla\Ff\bxi_{1}\right\rangle }{\sqrt{\lambda_{1}}},
	\]
	whichi is the PRA formula \eqref{eq:sinPRA}.

	\section{Total Lagrangian rotation in planar flows}
	\label{app:total_rot} 
	The polar rotation $\theta_{t_{0}}^{t}$ defined
	in Definition~\ref{def:PRA} is the net rotation of the $\{\bxi_{1},\bxi_{2}\}$
	eigenbasis over the time interval $[t_{0},t].$ This quantity, however,
	measures the rotation modulo $2\pi$ and does not differentiate between
	rotation by $\theta_{0}$ and $\theta_{0}+2k\pi$. Here, we also derive
	an expression for the total Lagrangian rotation of the eigenbasis
	that distinguishes between rotations differing by an integer multiple
	of $2\pi$. 
	
	Consider the equations of variations for a given infinitesimal displacement
	$\pmb\xi$, 
	\begin{equation}
	\dot{\pmb\xi}(t)=\bnabla\vc u(\vc x(t),t)\pmb\xi(t).\label{eq:eqvari}
	\end{equation}
	Write $\pmb\xi(t)=r\vc e_{\phi}$ where $\vc e_{\phi}=(\cos\phi,\sin\phi)^{\top}$
	and $(r,\phi)$ are functions of time. Substituting this in the equations
	of variations \eqref{eq:eqvari} we get 
	\begin{equation}
	\dot{r}\vc e_{\phi}+r\dot{\phi}\vc e_{\phi}^{\perp}=r\bnabla\vc u(\vc x(t),t)\vc 
	e_{\phi},
	\end{equation}
	with $\vc e_{\phi}^{\perp}=(-\sin\phi,\cos\phi)^{\top}$. Since $\vc e_{\phi}$
	and $\vc e_{\phi}^{\perp}$ are perpendicular, we have \begin{subequations}
		\begin{equation}
		\frac{\dot{r}}{r}=\langle\vc e_{\phi},\bnabla\vc u(\vc x(t),t)\vc e_{\phi}\rangle,
		\end{equation}
		\begin{equation}
		\dot{\phi}=\langle\vc e_{\phi}^{\perp},\bnabla\vc u(\vc x(t),t)\vc 
		e_{\phi}\rangle.\label{eq:thetaDot}
		\end{equation}
	\end{subequations} Therefore, solving Eq. \eqref{eq:thetaDot}, the
	total rotation of an arbitrary displacement vector 
	$\pmb\xi_{0}=(\cos\phi_{0},\sin\phi_{0})^{\top}$
	is given by 
	\begin{equation}
	\theta_{\mathrm{tot}}:=\phi(t)-\phi_{0}=\int_{t_{0}}^{t}\langle\vc 
	e_{\phi(\tau)}^{\perp},\bnabla\vc u(\vc x(\tau),\tau)\vc 
	e_{\phi(\tau)}\rangle\id\tau.\label{eq:totRot}
	\end{equation}
	If the initial vector $\pmb\xi_{0}$ is chosen to be $\bxi_{1}$ (or
	$\bxi_{2}$), $\theta_{\mathrm{tot}}$ measures the total rotation
	of the eigenbasis $\{\bxi_{1},\bxi_{2}\}$. We refer to $\theta_{\mathrm{tot}}$
	as the \emph{total Lagrangian rotation}.
	
	In practice, for evaluating the total Lagrangian rotation \eqref{eq:totRot},
	one needs to first compute the deformation gradient $\bnabla\Ff$
	from which the strain directions $\{\bxi_{1},\bxi_{2}\}$ are computed.
	The orientation of $\bxi_{1}$ (or alternatively $\bxi_{2}$) determines
	the appropriate initial condition $\vc 
	e_{\phi_{0}}=(\cos\phi_{0},\sin\phi_{0})^{\top}$
	with which Eq. \eqref{eq:thetaDot} should be solved. Note that Eq.
	\eqref{eq:thetaDot} must be solved simultaneous with the dynamical
	system $\dot{\vc x}=\vc u(\vc x,t)$ since $\bnabla\vc u$ is evaluated
	along trajectories $\vc x(t;t_{0},\vc x_{0})$.
	
	Therefore, evaluating the total Lagrangian rotation is more expensive
	than computing the PRA. The connected components of the level sets
	of $\theta_{\mathrm{tot}}$ and $\theta_{t_{0}}^{t}$ are identical
	by an argument similar to the one used in the proof of 
	Proposition~\ref{prop:connComp}.
	Thus the polar LCSs revealed by these two scalars are also 
	identical.

	\section{Proof of Proposition \ref{prop:obj}}\label{app:proof_prop3}
	Differentiating both sides of the formula~\eqref{eq:Eucl} with respect
	to the initial condition $\mathbf{x}_{0}$ gives
	\begin{equation}
	\bnabla\Ff=\mathbf{Q}(t)\mathbf{\bnabla\tilde{F}}_{t_{0}}^{t}\mathbf{Q}^{\top}(t_{0}),\label{eq:transfgrad}
	\end{equation}
	where $\bnabla\mathbf{\tilde{F}}_{t_{0}}^{t}$ denotes the deformation
	gradient in the $\mathbf{y}=\mathbf{\tilde F}_{t_0}^t(\vc y_0)$ coordinate system. 
	From 
	\eqref{eq:transfgrad},
	we obtain
	\begin{align*}
	\mathbf{\nabla\tilde{F}}_{t_{0}}^{t}
	&= \mathbf{Q}^{\top}(t)\bnabla\Ff\mathbf{Q}(t_{0})\\
	&=\mathbf{Q^{\top}}(t)\mathbf{R}_{t_{0}}^{t}\mathbf{U}_{t_{0}}^{t}\mathbf{Q}(t_{0})\\
	&=\mathbf{Q^{\top}}(t)\mathbf{R}_{t_{0}}^{t}\mathbf{Q}(t_{0})
	\mathbf{Q^{\top}}(t_{0})\mathbf{U}_{t_{0}}^{t}\mathbf{Q}(t_{0})\\
	&=\mathbf{\tilde{R}}_{t_{0}}^{t}\mathbf{\tilde{U}}_{t_{0}}^{t},
	\end{align*}
	where the rotation tensor 
	$\mathbf{\tilde{R}}_{t_{0}}^{t}\mathbf{=Q^{\top}}(t)\mathbf{R}_{t_{0}}^{t}\mathbf{Q}(t_{0})$
	and the positive definite, symmetric tensor 
	$\mathbf{\tilde{U}}_{t_{0}}^{t}=\mathbf{Q^{\top}}(t_{0})\mathbf{U}_{t_{0}}^{t}\mathbf{Q}(t_{0})$
	represent the unique polar decomposition of $\mathbf{\nabla\tilde{F}}_{t_{0}}^{t}$.
	Then 
	\begin{align*}
	\tr\mathbf{\tilde{R}}_{t_{0}}^{t}(\mathbf{y}_{0}) & 
	=\tr\left[\mathbf{Q^{\top}}(t)\mathbf{R}_{t_{0}}^{t}(\mathbf{x}_{0})\mathbf{Q}(t_{0})\right]\\
	\\
	& =\tr\begin{pmatrix}
	\cos\left[\theta_{t_{0}}^{t}(\mathbf{x}_{0})+q(t_{0})-q(t)\right] & 
	-\sin\left[\theta_{t_{0}}^{t}(\mathbf{x}_{0})+q(t_{0})-q(t)\right]\\
	\sin\left[\theta_{t_{0}}^{t}(\mathbf{x}_{0})+q(t_{0})-q(t)\right] & 
	\cos\left[\theta_{t_{0}}^{t}(\mathbf{x}_{0})+q(t_{0})-q(t)\right]
	\end{pmatrix}\\
	& =2\cos\left[\theta_{t_{0}}^{t}(\mathbf{x}_{0})+q(t_{0})-q(t)\right],
	\end{align*}
	where $q(t)$ represents the angle of rotation associated with $\mathbf{Q}(t)$.
	Therefore, if the polar rotation angle generated by transformed rotation
	tensor $\mathbf{\tilde{R}}_{t_{0}}^{t}$ is 
	$\tilde{\theta}_{t_{0}}^{t}(\mathbf{y}_{0})$, 
	then
	\begin{align}
	\cos\left(\tilde{\theta}_{t_{0}}^{t}(\mathbf{y}_{0})\right) & 
	=\frac{1}{2}\tr\mathbf{\tilde{R}}_{t_{0}}^{t}(\mathbf{y}_{0})\nonumber\\
	& =\cos\left(\theta_{t_{0}}^{t}(\mathbf{x}_{0})+q(t_{0})-q(t)\right).
	\label{eq:rpa_rot}
	\end{align}
	
	Consequently, if two points $\mathbf{x}_{0}$ and $\mathbf{\hat{x}}_{0}$
	lie on the same connected level set of $\theta_{t_{0}}^{t}(\mathbf{x}_{0})$,
	then the corresponding points also lie on a connected level set of
	$\tilde{\theta}_{t_{0}}^{t}(\mathbf{y}_{0})$, even though we generally
	have 
	$\theta_{t_{0}}^{t}(\mathbf{x}_{0})\neq\tilde{\theta}_{t_{0}}^{t}(\mathbf{y}_{0})$.
	
	We note that the level sets of PRA in three dimensions are generally not
	objective. An 
	essential part of the above argument, leading to equation~\eqref{eq:rpa_rot}, is that 
	the 
	rotation matrices $\vc Q(t)$, $\R(\vc x_0)$ and $\vc Q(t_0)$ share the same axis of 
	rotation (i.e., the normal to the plane of motion). In three dimensions, such a 
	uniform 
	axis of rotation does not generally exist, and hence a relation similar to 
	\eqref{eq:rpa_rot} does not hold.
	
\end{appendices}

%\bibliographystyle{unsrtnat} %plainnat
%\bibliography{../bibliog}

\end{document}